\pdfoutput=1
\documentclass[10pt, letterpaper]{IEEEtran}
\IEEEoverridecommandlockouts
\usepackage{amsmath,graphicx,amssymb,amsthm}% ,epstopdf,spconf,amssymb,epsfig,fullpage}%, ,graphicx,psfrag,url,amsmath,amsthm,comment}
\usepackage{cite}
\usepackage{dsfont}
\usepackage{url}
\usepackage{color}
\usepackage{epstopdf}
\usepackage{verbatim}
\usepackage{cuted}
\usepackage{lipsum,graphicx,subcaption}
\usepackage[utf8]{inputenc}
\usepackage[english]{babel}
\usepackage{caption}
\usepackage{enumitem}
\usepackage{algorithm}
\usepackage{algorithmic}
\DeclareCaptionLabelFormat{lc}{\MakeLowercase{#1}~#2}
\usepackage{balance}
\usepackage{float}
\usepackage{mathtools}

\DeclarePairedDelimiter\floor{\lfloor}{\rfloor}

\newtheorem{assumption}{Assumption}
\newtheorem{theorem}{Theorem}

\newtheorem{definition}{Definition}
\newtheorem{proposition}{Proposition}
\newtheorem{corollary}{Corollary}

\DeclareMathOperator*{\argmax}{arg\,max}
\DeclareMathOperator*{\argmin}{arg\,min}
\allowdisplaybreaks
\renewcommand{\figurename}{Fig.}
\usepackage[protrusion=true,expansion=true]{microtype}
\usepackage[font=small]{caption}
\let\emptyset\varnothing

\usepackage{tabularx,booktabs}
\newcolumntype{Y}{>{\centering\arraybackslash}X}
\usepackage{multirow}

\def\BibTeX{{\rm B\kern-.05em{\sc i\kern-.025em b}\kern-.08em
     T\kern-.1667em\lower.7ex\hbox{E}\kern-.125emX}}

\begin{document}
% \title{Device Sampling for Heterogeneous Federated Learning: Theory, Algorithms, and Implementation} 
\title{Device Sampling and Resource Optimization for Federated Learning in Cooperative Edge Networks}  %, D2D Communication, 
% \vspace{-0.1in}}

% Smart Device Sampling with D2D Offloading for Federated Learning over Heterogeneous Networks
% \author{\semiSmall{Su Wang, Mengyuan Lee, Seyyedali Hosseinalipour, Roberto Morabito, Mung Chiang, and Christopher G. Brinton}
\author{\IEEEauthorblockN{\normalsize Su Wang\IEEEauthorrefmark{1}, Roberto Morabito\IEEEauthorrefmark{3}, Seyyedali Hosseinalipour\IEEEauthorrefmark{4}, Mung Chiang\IEEEauthorrefmark{1}, and Christopher G. Brinton\IEEEauthorrefmark{1}\\}
\IEEEauthorblockA{\small \IEEEauthorrefmark{1}School of Electrical and Computer Engineering, Purdue University\\
% \IEEEauthorrefmark{2}College of Information Science and Electronic Engineering, Zhejiang University\\
\IEEEauthorrefmark{3}Communication Systems Department, EURECOM\\
\IEEEauthorrefmark{4}Department of Electrical Engineering, University at Buffalo--SUNY\\
Email: \IEEEauthorrefmark{1}\{wang2506, chiang, cgb\}@purdue.edu,
% \IEEEauthorrefmark{2}mengyuan\textunderscore lee@zju.edu.cn,
\IEEEauthorrefmark{3}roberto.morabito@eurecom.fr,
\IEEEauthorrefmark{4}alipour@buffalo.edu}
\vspace{-0.2in}
\thanks{{This work is an extension of our prior work titled ``Device sampling for heterogeneous federated learning: Theory, algorithms, and implementation"~\cite{wang2021device} published in the Proceedings of IEEE INFOCOM, 2021.} }
\thanks{{Christopher G. Brinton is supported by the National Science Foundation (NSF) under grants CNS-2146171 and CPS-2313109, and the Office of Naval Research (ONR) under grants N000142212305 and N0001423C1016.}}}
% \vspace{-0.4in}}

\maketitle
\begin{abstract}
% \question{The currently abstract length is 290 words. TBD: cut to 200 words.}
% Federated learning (FedL) has attracted significant attention as an architecture for distributed machine learning (ML).
The conventional federated learning (FedL) architecture distributes machine learning (ML) across worker devices by having them train local models that are periodically aggregated by a server. FedL ignores two important characteristics of contemporary wireless networks, however: (i) the network may contain heterogeneous communication/computation resources, and (ii) there may be significant overlaps in devices' local data distributions. In this work, we develop a novel optimization methodology that jointly accounts for these factors via intelligent device sampling complemented by device-to-device (D2D) offloading. 
Our optimization methodology aims to select the best combination of sampled nodes and data offloading configuration to maximize FedL training accuracy while minimizing data processing and D2D communication resource consumption subject to realistic constraints on the network topology and device capabilities. 
Theoretical analysis of the D2D offloading subproblem leads to new FedL convergence bounds and an efficient sequential convex optimizer. Using these results, we develop a sampling methodology based on graph convolutional networks (GCNs) which learns the relationship between network attributes, sampled nodes, and D2D data offloading to maximize FedL accuracy. 
Through evaluation on popular datasets and real-world network measurements from our edge testbed, we find that our methodology outperforms popular device sampling methodologies from literature in terms of ML model performance, data processing overhead, and energy consumption. %while sampling less than 5\% of all devices outperforms conventional FedL substantially both in terms of trained model accuracy and required resource utilization.}

\end{abstract}
\section{Introduction}

\noindent  
The proliferation of smartphones, unmanned aerial vehicles (UAVs), and other devices comprising the Internet of Things (IoT) is causing an exponential rise in data generation and large demands for machine learning (ML) at the edge~\cite{8373692,alsheikh2022five}. %,8270639}. 
For example, sensor and camera modules on self-driving cars produce up to 1.4 terabytes of data per hour~\cite{carData1} which enable the training of ML models for intelligent navigation and autonomous driving~\cite{wu2022intelligence}. 
Unfortunately, traditional ML techniques where model training is conducted at a centralized server are not applicable to such distributed environments as the data often cannot be transferred to a server. 
There are two challenges in particular: (i) given the large volumes of data at the network edge, data transfer to a server imposes long delays and overloads the network infrastructure~\cite{li2022unified}, and (ii) in certain applications (such as healthcare), some users are unwilling to share their data owing to privacy concerns~\cite{gupta2023fedcare}. 
% The traditional paradigm in ML of centralized training at a server is often not feasible in such environments since (i) transferring these large volumes of data from the devices to the cloud imposes long transfer delays and (ii) users are sometimes unwilling to share their data due to privacy concerns~\cite{7498684}. % ,carData2}

%\subsection{Federated Learning and Challenges}
Federated learning (FedL) has become a popular distributed ML technique aiming to overcome these challenges~\cite{mcmahan2017communication,konevcny2016federated}. Under FedL, devices train models on their local datasets, typically by means of gradient descent, and a server periodically aggregates the parameters of local models to form a global model. These global model parameters are then transferred back to the devices for the next round of local updates, as depicted in Fig.~\ref{fig:simpleFL}.
% The modern era has witnessed an explosion in the number of internet of things (IoT) devices, such as cellular phones, tablets, smart cars, unmanned aerial vehicles (UAVs), and a considerable technological advancement in the mounted equipment of such devices \cite{7498684}. IoT devices are currently equipped with miscellaneous sensors and camera modules capable of data gathering for different purposes. For instance, auto-driving technology is continuously improved using RADAR, LIDAR, camera, ultrasonic, GNSS, and IMU modules that produce up to 1.4 terabyte of data per hour~\cite{carData1,carData2}. One of the emerging use cases of the IoT devices' collected data is training machine learning (ML) models~\cite{8373692,8270639}. 
% Elegant framework of federated learning (FedL) has emerged to overcome these challenges. FedL distributes model training across the edge devices who perform local model updates, usually by means of a gradient descent method, and the server periodically aggregates the models to form a new global model, which is transferred back to the edge devices for the next round of local updates. Using this method model training is conducted while users' datasets are remained local~\cite{mcmahan2017communication,konevcny2016federated}. A schematic of the learning architecture of FedL is depicted in Fig.~\ref{fig:simpleFL}.
In standard FedL methodologies, each device processes its own collected data and operates independently within an aggregation period~\cite{wang2021novel,lu2022class}. This will, however, become problematic in terms of upstream device communication and local device processing requirements as implementations scale to networks consisting of millions of heterogeneous wireless devices~\cite{niknam2020federated,hosseinalipour2020federated,zehtabi2022decentralized}. 
%With each device processing its collected data, FedL assumes that the data distributions are non-i.i.d., and each device operates independently within an aggregation period. This will pose scalability issues, as FedL is expected to be implemented over wireless networks consisting of millions/billions of devices~\cite{niknam2020federated,hosseinalipour2020federated}. %zhu2020towardOn the other hand,

At the same time, device-to-device (D2D) communications in cooperative 5G/IoT edge networks can enable \textit{consensual} local offloading of data processing from resource hungry to resource rich devices~\cite{chang2023bev,wang2023towards}. %that are becoming part of 5G and IoT
Additionally, we can expect that for particular applications, the datasets collected across devices will contain varying degrees of similarity, e.g.,
%,verma2019approaches
% In fog applications such as internet-connected vehicles and industrial IoT \cite{cisco-5giiot}, empirical similarity among geographically close devices is likely. 
%from the environment by smart cars or cellular phones that are in close proximity, or the information produced 
images gathered by UAVs surveying the same area~\cite{9084352,wang2022uav}. Processing similar data distributions at multiple devices adds unnecessary computational and communication overheads to FedL, thus leading to an opportunity for efficiency improvement. 
%Exploiting such dependencies in these scenarios can alleviate the amount of processing and upstream communication required from the devices. %One challenge to doing so, however, is that network operators/servers seldom have exact or detailed information about the network, the devices, and their relationships. 
%Motivated by this, we pose the following research question: \textit{How can overlaps that exist across local device datasets be exploited to reduce the FedL resource demand?}
%with minimal global knowledge of devices' local data distributions 

\begin{figure}[t]
\includegraphics[width=.45\textwidth]{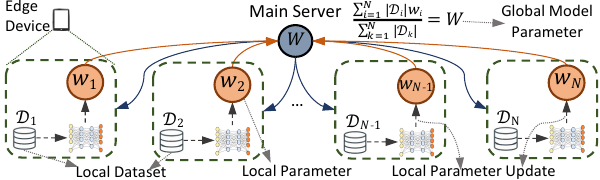}
\centering
\caption{Architecture of conventional federated learning (FedL).}
\label{fig:simpleFL}
\vspace{-6mm}
\end{figure}

%, i.e., determining the subset of devices that maximizes the expected contribution to training subject to resource constraints , in this paper,
Motivated by this, we develop a novel methodology for \textit{smart device sampling with data offloading} in FedL. Specifically, we formulate a joint sampling and data offloading optimization problem where devices expected to maximize contribution to model training are sampled for training participation, while devices that are not selected may transfer data to devices that are. 
This data offloading is motivated by paradigms such as \textit{fog learning}~\cite{hosseinalipour2020federated,hosseinalipour2023parallel}, in which data offloading only occurs among trusted devices; devices that have privacy concerns are exempt from data offloading. 
Furthermore, the data offloading is performed according to estimated data dissimilarities between nodes, which are updated as transfers occur. 
We show that our methodology yields comparable performance to conventional FedL, and superior performance relative to baselines from literature while significantly reducing network resource utilization.

% \vspace{-1mm}
\subsection{Related Work}
\textbf{Resource efficiency in federated learning}
Since conventional FedL in large-scale networks can incur heavy communication and computation resource use, many works have investigated avenues to reduce the resource burden of edge devices in FedL. In particular, a few popular and recent approaches have focused on efficient encoding designs to reduce parameter transmission sizes~\cite{9155479,sattler2019robust}, optimizing the frequency of global aggregations~\cite{sun2020adaptive,wang2019adaptive,wu2023joint}, and device sampling \cite{ji2020dynamic,karimireddy2020scaffold,pmlr-v151-jee-cho22a}. 
Our work falls into the device sampling category, in which methodologies rely on intelligent selection of devices for the global aggregation stage of FedL to improve resource efficiency. 
In this regard, most works have assumed a static or uniform device selection strategy, e.g.,~\cite{wang2019adaptive,yang2019energy,tran2019federated,konevcny2016federated,sahu2018convergence,reisizadeh2020fedpaq,ji2020dynamic}, 
where the main server chooses a subset of devices either uniformly at random or according to a pre-determined sampling distribution. 
% We compare our methodology against a few of these more recent/popular baselines in our experimental evaluation section. % These static selection techniques enable communication savings by reducing device-to-server model transmissions as well as computational savings owing to fewer devices that need to perform ML model training.}  

% There is also an emerging line of work on device sampling based on wireless channel characteristics, specifically in cellular networks~\cite{8851249,shi2019device,xia2020multi}. 
% By contrast, we develop a sampling technique that adapts to the heterogeneity of device resources and overlaps across local data distributions that are key characteristics of contemporary wireless edge networks. 
While there are many different approaches for device sampling for FedL such as sampling based on wireless channel characteristics in cellular networks (e.g.,~\cite{8851249,shi2019device,xia2020multi,ren2020scheduling}), we focus on developing a sampling technique that adapts to the heterogeneity of device resources and the similarities/dissimilarities of data distributions across large-scale contemporary wireless edge networks. 
When investigating literature in this line of device sampling for FedL, popular works have attempted device sampling based on each device's instantaneous contributions to global updates~\cite{9155494}, data proportional sampling based on devices' local dataset sizes~\cite{li2019convergence,karimireddy2020scaffold}, and sampling devices based on local ML training loss~\cite{pmlr-v151-jee-cho22a}.
Specifically, when compared to the literature on device sampling by each device's instantaneous contributions to the global updates~\cite{9155494,pmlr-v151-jee-cho22a}, our proposed methodology introduces a novel perspective based on device data similarities, and furthermore exploits the proliferation of D2D cooperation at the wireless edge~\cite{su2021secure,wang2023towards,wang2021network,lin2021friend} to diversify each selected device's local data via D2D offloading. Our work thus considers the problems of sampling and D2D offloading for FedL jointly, and leads to new analytical convergence bounds and algorithms used by implementations.
% Specifically, when compared to the limited literature on device sampling by each device's instantaneous contributions to global updates~\cite{9155494}, we introduce a novel perspective based on device data similarities. 
% Our methodology exploits the proliferation of D2D communications at the wireless edge~\cite{hosseinalipour2020multi}, to diversify each selected device's local data via D2D offloading. Our work thus considers the novel problems of sampling and D2D offloading for FedL jointly, and leads to new analytical convergence bounds and algorithms used by implementations. 

% It is worth mentioning two parallel lines of work in FedL that consider relationships between node data distributions. One is on fairness \cite{williamson2019fairness}, in which the objective is to train the ML model without biasing the result towards any one device's distribution, e.g.,~\cite{mohri2019agnostic,li2019fair}. Another leverages transfer learning techniques~\cite{pan2009survey} to build models across data parties (e.g., companies or enterprises) that possess partial overlaps~\cite{li2019fedmd,wang2023multi,gao2019privacy}. Our work is focused on a fundamentally different objective, i.e., network resource efficiency optimization. 
% % through device sampling and D2D data offloading. %yang2019federated,

\textbf{Towards further cooperation in federated learning.}
Most works in advancing conventional FedL~\cite{tran2019federated,konevcny2016federated,dou2023device,zhang2023semi,qu2021efficient} limit network cooperation to the global aggregation stages. 
While this is effective in situations requiring absolute data privacy, many practical large-scale and mobile edge networks are not limited to such an extent~\cite{xie2020fast,xia2021ol,wang2023towards}. 
Especially in large-scale and mobile edge networks, devices can be heterogeneous with respect to their data privacy demands~\cite{li2023energy,savazzi2021opportunities,lin2022relay}. 
While some devices may require absolute data privacy, many devices may be willing to share non-sensitive data or share all of their data to a device that they trust. 
Such is the case in recent literature on connected autonomous driving~\cite{barbieri2022decentralized,chang2023bev}, which rely on inter-device cooperation for normal functionality and is now attempting to integrate the benefits of modern ML. 
With this in mind, works such as~\cite{li2023energy,ganguly2023multi,guo2022hybrid} have begun expanding the scope of cooperation in FedL. Specifically,~\cite{li2023energy} investigates FedL performance in the presence of D2D model offloading among edge devices,~\cite{ganguly2023multi} proposes device to edge server data offloading so that edge servers can contribute to FedL training, and~\cite{guo2022hybrid} leverages D2D cooperation in edge networks to develop a modified version of SGD for FedL specifically. %novel ML training process for FedL.
Our work contributes to existing literature in cooperative FedL by investigating how D2D data cooperation can lead to optimized device sampling with respect to ML training efficacy and total network resource consumption. Furthermore, our work assumes that devices are heterogeneous with respect to their privacy demands (i.e., they trust some edge devices but not others), resulting in some D2D links enabled and others disabled. %with some D2D links enabled and others disabled. 

\begin{figure}[t]
\includegraphics[width=.48\textwidth]{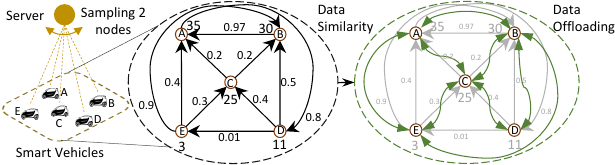}
\centering
\caption{A motivating example of a wireless network composed of 5 connected vehicles and an edge server. The server can only sample two vehicles to participate in FedL training.}
\label{fig:SimilarOffload}
\vspace{-5mm}
\end{figure}
% where only two can be sampled in FedL image classification. Smart device sampling should consider both the data similarity graph and the heterogeneous capabilities of nodes in data offloading.
\vspace{-1mm}
\subsection{Motivating Toy Example}\label{sec:toy}
% Consider a scenario in which the network operator can exploit a certain fraction of the available devices to conduct FedL, each of which has a dataset consisting of a set of features and labels.\footnote{The limitation in the number of the sampled nodes can be due to many factors, such as available bandwidth or network congestion.}
% The notion of similarity between the devices' datasets can be captured via many metrics, such as applying the Kullback–Leibler (KL)-divergence or other probabilistic distance measures~\cite{coeurjolly2007normalized} on the devices' data distributions. Given the fact that obtaining these measures is very hard in practice where the dataset of the users are not observable by the network operator, the measure of similarity can be as simple as a guess on the probability of existence of similar data points between the nodes. In fact, our developed method is  general and is not sensitive to the specific definition of similarity. 
Consider Fig.~\ref{fig:SimilarOffload}, wherein five heterogeneous connected vehicles communicate with an edge server to train an object detection model. Due to limited bandwidth, the server can only exploit $2$ out of the $5$ vehicles to conduct FedL training, but needs to train a model representative of the entire dataset within this network. %of images across the cars. 
% In Fig.~\ref{fig:SimilarOffload}, we consider a weighted network graph where the normalized edge weights capture the similarity between the datasets of the nodes (a higher weight implies a higher similarity), where nodes posses different number of data points. 
The computational capability of each vehicle, i.e., the number of processed datapoints in one aggregation period, is shown next to the node representing the vehicle, and the edge weights in the data similarity graph capture the similarity between the local data of the vehicles. % at the present time.
% For instance, the data collected by sensing equipment on two smart cars on the same road
% similar commutes will hthe information on which can be obtained via the commonly used navigation software such as Google Maps and Waze~\cite{jeske2013floating},  located in the same area is expected to be similar. 
%For instance, the environmental data collected by two smart cars with
%similar commutes, %the information on 
%which can be obtained via navigation software such as Google Maps and Waze~\cite{jeske2013floating}, is expected to be highly similar.
%The same holds for two smartphones with close geographical proximity.
%a smartphone in a certain area is expected to be similar to the data gathered by a closely located smartphone. The same holds for two smart cars that have similar commute rounds, the information on which can be obtained via the commonly used navigation software such as Google Maps and Waze~\cite{jeske2013floating}.
%Kullback-Leibler (KL)-divergence or other 
Rather than using statistical distance metrics~\cite{liese2006divergences}, which are hard to compute in this distributed scenario, the data similarities could be estimated by commute routes and geographical proximity~\cite{jeske2013floating}. Also, in D2D-enabled environments, nodes can exchange small data samples with trusted neighbors to calculate similarities locally and report them to the server.
%Our method is general and not sensitive to the specific definition of similarity.
% As the dataset of the devices are mostly gathered via their sensing equipment and interactions with the end users, one can expect to have certain degrees of similarity between the data of different devices. 

In Fig.~\ref{fig:SimilarOffload}, if the server samples the vehicles with the highest computational capabilities, i.e., $A$ and $B$, the sampling is inefficient due to the high data similarity between them. Conversely, if it samples those with the lowest similarity, i.e., $D$ and $E$, the local models will be based on low computational capabilities, which will often result in a low accuracy (and could be catastrophic in this vehicular scenario). Optimal sampling of the vehicles considering both data similarities and computational capabilities is thus critical to the operation of FedL.

% As the tradeoff between the accuracy of the ML model and the number of data points and the quality of data (number of labels engaged in training, the dissimilarity of the features, the distribution of the features) is not yet known and it is intractable to derive, there is no obvious choice for the network operator even in this small scale example. This reveals that the sampling of the devices considering their data similarity is highly non-triviality.

% We take one step further and consider a more interesting and even more challenging scenario, in which after device sampling some nodes can engaged in D2D data offloading to the sampled nodes.
We take this one step further to consider how D2D offloading can lead to augmented local distributions of sampled vehicles.
% This can be exploited in an efficient manner to compensate for the data-points that the sampled nodes are lacking.
The node sampling must consider the neighborhoods of different vehicles and the capability of data offloading in those neighborhoods. 
For example, D2D is cheaper in terms of resource utilization among vehicles that are close in proximity. 
The feasible offloading topology shown in Fig.~\ref{fig:SimilarOffload} is represented by the data offloading graph. Given $C$'s high processing capability and data dissimilarity with neighboring vehicles $E$ and $D$, sampling $C$ in a D2D-optimized solution combined with data offloading from $E$ and $D$ to $C$ can yield a composite of all three vehicles' distributions. The purpose of this paper is to model these relationships for general edge/fog networks and optimize the resulting sampling and offloading configurations.

\vspace{-1.1mm}
\subsection{Outline and Summary of Contributions}
% We outline the novelty and contributions of this work below:\\

\begin{itemize}[leftmargin=4mm]
    % \item  We advocate network-aware FedL that encapsulates the underlying network structure among the devices into the learning paradigm. This opens door to a series of new opportunities and challenges. In this work, we focus on exploiting our proposed framework to conduct FedL under resource constraints, where a very limited set of devices can engage in uplink transmissions to the main server. 
    % \item We formulate a new problem concerned with smart device sampling with embedded D2D data offloading aiming to maximize the ML model precision in resource constrained environments. We demonstrate that the problem is in general highly non-trivial and intractable.
    % \item 
    %   We theoretically obtain the upper bound of convergence of FedL under arbitrary device sampling and data offloading strategy when multiple rounds of local updates are conducted in between consecutive global aggregations. 
    %  We then propose a method that \textit{learns} the intractable trade-off between the number of data points at different devices and respective data similarities to the accuracy of FedL. It further \textit{encapsulates} the network structure and the offloading scheme to infer the best sampled nodes to maximize the FedL model precision.

% \item We formulate the joint sampling and D2D offloading optimization problem 
\item We formulate the joint sampling, D2D offloading, and communication and computation resource optimization problem for maximizing FedL model accuracy subject to realistic network constraints (Sec.~\ref{s:sm}). 
%resource
%Given its non-convex mixed integer program categorization, we investigate D2D offloading and device sampling as two subproblems.

\item We develop an approach to D2D data similarity estimation by comparing the dataset centroids across network devices (Sec.~\ref{s:sm}). This method enables the server to perform the optimization without harvesting actual device data.
% \item {\color{blue} centroid based D2D cooperation} 

\item Our theoretical analysis of the offloading subproblem for a fixed sampling strategy yields a new upper bound on the convergence of FedL under an arbitrary data sampling strategy (Sec.~\ref{s:p1}). Using this bound, we derive an efficient sequential convex optimizer for the offloading strategy.

\item We propose a novel ML-based methodology that learns the desired combination of sampling and resulting offloading (Sec.~\ref{s:p2}). We encapsulate the network structure and offloading scheme into model features and learn a mapping to the sampling strategy that maximizes expected FedL accuracy.

\item We evaluate our methodology through experiments with network parameters obtained from our testbed of wireless IoT devices (Sec.~\ref{s:numRes}). 
Our methodology yields better ML accuracies and less resource consumption than recent device sampling baselines from literature. %popular device sampling methodologies from recent literature.

\item We experimentally characterize our formulation by investigating its behavior with respect to various key optimization variables while varying network and sampling device set sizes (Sec.~\ref{s:numRes}).
% that exceed FedL trained on all devices with significant reductions in processing requirements. 
%%on real-world ML tasks with 

\end{itemize}

\vspace{-1mm}
\section{System and Optimization Model} \label{s:sm}

%\ali{i) you need to define both the similarity and the cost of offloading..... Each of them is a seperate graph. We "use" stragglers in D2D. Devices are partitioned into two sets: i)those can be sampled (formulation..), ii) those that cannot (stragglers). clearly explain all the devices characteristics and why they are functions of time? (like capacity can decrease since the battery decrease ....) --- Motivate the change of edge weight over time very well (that is a key contribution : like you avoid transferring between the devices that have already transferred a lot). Mention the fact that the similarity between the nodes can be defined in "multiple ways". Also, the  edge weights are directional and not reciprocal (and say exactly why!). You need to clearly explain, data distribution and data realization. 

%ii) fixed network graph, device capabilities are fixed within each aggregation iterval}

\noindent In this section, we formulate the joint sampling and offloading optimization (Sec.~\ref{ss:ovr_problem}). We first introduce our edge device (Sec.~\ref{sss:devices}), network topology and data similarity (Sec.~\ref{sss:graph}), and ML (Sec.~\ref{ss:mlp}) system models.

% used in the optimization. %and finally present the optimization formulation (Sec. \ref{ss:ovr_problem}).

\subsection{Edge Device Model} \label{sss:devices}
We consider a set of devices $\mathcal{N} = \{1,\cdots,N\}$ connected to a server, and time span $t=0,\cdots,T$ for model training. Each device $i \in \mathcal{N}$ possesses a data processing capacity $P_i(t) \geq 0$, which limits the number of datapoints it can process for training at time $t$, and a unit data processing cost $p_i(t) \geq 0$. Intuitively, $p_i(t)$, $P_i(t)$ are related to the total CPU cycles, memory (RAM), and power available at device $i$~\cite{morabito2018legiot}. 
These factors are heterogeneous and time-varying, e.g., as battery power fluctuates and as each device becomes occupied with other tasks. Additionally, for each $i \in \mathcal{N}$, we define $\Psi_{i}(t) > 0$ as the data transmit budget, and $\psi_{i,j}(t) > 0$ as the unit data transmission cost from device $i$ to device $j$. %across devices, and also vary over time $t$
Intuitively, $\psi_{i,j}(t)$, $\Psi_{i}(t)$ are dependent on factors such as the available bandwidth and distance. For example, devices that are closer in proximity would be expected to have lower $\psi_{i,j}(t)$.
Our specific derivations of these data processing and communication variables are based on empirical measurements from our IoT testbed, which we describe in detail in Sec.~\ref{ss:setup}.

Due to resource constraints, the server selects a set $\mathcal{S} \subseteq \mathcal{N}$ of devices to participate in FedL training. Some devices $i \in \hat{\mathcal{S}} \triangleq \mathcal{N} \setminus \mathcal{S}$ may be stragglers, i.e., possessing insufficient $P_i(t)$ to participate in training, but nonetheless gather data. Different from most works, our methodology will seek to leverage the datasets captured by nodes in the unsampled set $\hat{\mathcal{S}}$ via local D2D communications with nodes in the sampling set $\mathcal{S}$.

We denote the dataset at device $i$ for the specific ML application by $\mathcal{D}_i(t)$. $\mathcal{D}_i(0)$ is the initial data at $i$, which evolves as offloading takes place. Henceforth, we use calligraphic font (e.g., $\mathcal{D}_i(t)$) to denote a set, and non-calligraphic (e.g., $D_i(t) = |\mathcal{D}_i(t)|$) to denote its cardinality. Each data point $d \in \mathcal{D}_i(t)$ is represented as $d=(\mathbf{x}_d,y_d)$, where $\mathbf{x}_d\in \mathbb{R}^M$ is a feature vector of $M$ features, and $y_d\in\mathbb{R}$ is the target label.
\vspace{-1mm}
\subsection{Network Topology and Data Similarity Model} \label{sss:graph}
We consider a time-varying network graph $G = (\mathcal{N},\mathcal{E}(t))$, among the set of nodes $\mathcal{N}$ to represent the available D2D topology. 
Here, $\mathcal{E}(t)$ denotes the set of edges or connections between the nodes, where $(i,j)\in\mathcal{E}(t)$ if node $i$ is able/willing to transfer data in D2D mode to node $j$ at time $t$.
This willingness to transfer data can depend on the trust between the devices as an example, and whether the devices are D2D-enabled. For instance, smart home peripherals can likely transfer data to their owner's smartphone, while different smartphones in an airport may be unwilling to share data. %certain smart cars in the vehicular network from Fig.~\ref{fig:SimilarOffload} may be unwilling to communicate. %~\cite{liwang2019allocation}.
We capture these potential D2D relationships using the adjacency matrix $\mathbf{A}(t) =[A_{i,j}(t)]_{1\leq i,j\leq N}$, where $A_{i,j}(t)=1$ if $(i,j)\in\mathcal{E}(t)$, and $A_{i,j}(t)=0$ otherwise.
%For $(i,j) \in \mathcal{E}(t)$, we then define $\Psi_{i,j}(t) > 0$ as the data transfer limit over this link, i.e., the number of datapoints that can be transferred during $t$, and $\psi_{i,j}(t) > 0$ as the unit data transmission cost. Intuitively, $\psi_{i,j}(t)$, $\Psi_{i,j}(t)$ are dependent on factors such as the available bandwidth, the channel interference conditions, and the transmit power available at device $i$. For example, devices that are closer in proximity would be expected to have higher $\Psi_{i,j}(t)$ and lower $\psi_{i,j}(t)$.

%\subsection{Data Similarity Model} \label{sss:data_model}

We define $\Phi_{i,j}(t) \in [0,1]$ as the fraction of node $i$'s data offloaded to node $j$ at time $t$. To optimize this, we are interested in the similarity among local datasets. We define the similarity matrix $\boldsymbol{\lambda}(t) \triangleq [\lambda_{i,j}(t)]_{1\leq i,j\leq N}$ among the nodes at time $t$, where $0\leq\lambda_{i,j}(t)\leq 1$. 
Lower values of $\lambda_{i,j}(t)$ imply a higher dataset similarity between nodes $i$ and $j$, and thus less offloading benefit. %, as discussed in Sec.~\ref{sec:toy},
In practice, neither the server nor the devices have exact knowledge of the local data distributions.

To this end, we consider an interpretation of similarity based on comparing data centroids across network devices. 
Each centroid is an average of the data contained within a data cluster. 
%which are commonly defined based on the number of unique categories of data (i.e., number of unique labels) contained within a dataset. 
% unique data labels within its local dataset $\mathcal{D}_i(t)$, 
{\color{black}In our setting, each device $i$ has $\mathcal{C}_i$ clusters, the contents of which are determined based on K-means~\cite{krishna1999genetic} clustering which is commonly used in existing literature~\cite{jacksi2020clustering,ribas2012similarity} to estimate dataset similarities.}
Each cluster $c \in \mathcal{C}_i$ contains the data $\mathcal{D}_i^{c}(t)$ so that $\mathcal{D}_i(t) = \cup_{c \in \mathcal{C}_i} \mathcal{D}_i^{c}(t)$. 
Thus, the $c$-th centroid $\mu_{i}^{c}(t)$ at device $i$ is computed as:
\begin{equation} \label{eq:centroid_computation}
    \mu_{i}^{c}(t) = \frac{1}{\vert \mathcal{D}_i^{c}(t) \vert} \sum_{d \in \mathcal{D}_i^{c}(t)} \mathbf{x}_d. 
\end{equation}
For each D2D link, we determine the similarity value $\lambda_{i,j}(t)$ based on the minimum total centroid difference for all clusters between devices $i$ and $j$.\footnote{\color{black}The goal of the data similarity estimation process is to gauge the data offloading value. If two device pairs have similar minimum centroid difference, then our formulation estimates that these device pairs have similar benefit from data offloading.}
% cluster matching and further processing. XXX - which can be done at the server - D2S transmissions}
Specifically, we first need to match each cluster $b \in \mathcal{C}_j$ to some cluster $c \in \mathcal{C}_i$ with the smallest centroid difference between them, i.e., we want 
\begin{equation} \label{eq:min_centroid_cluster}
    c = \argmin_{c \in C_i} \vert \sigma_{i,j}^{c,b}(t) \vert,
\end{equation}
where $\sigma_{i,j}^{c,b}(t) = \mu_i^c(t) - \mu_j^b(t)$ and $\sigma_{i,j}^{c,b}(t)=0$ implies that clusters $b$ and $c$ are essentially identical.
However,~\eqref{eq:min_centroid_cluster} enables a single cluster $c$ to be the matching cluster for multiple $b \in \mathcal{C}_j$, which fails to fully characterize the D2D similarity between all of the data at devices $i$ and $j$ (as its possible for $\mathcal{D}_i^{c}(t)$ to be the only subset of data that is characterized at this link).
% fails to fully characterize D2D data offloading opportunities from device $i$ to device $j$ (as at most $\mathcal{D}_i^c(t)$ is characterized). 
As a result, when a cluster $c \in \mathcal{C}_i$ has the smallest $\sigma_{i,j}^{c,b}(t)$ for multiple clusters $b \in \mathcal{C}_j$, we choose to link $c$ to the cluster $b$ with the smallest $\sigma^{c,b}_{i,j}(t)$, and subsequently use the heuristic $\widetilde{\mu}^{c}_i(t) \rightarrow \infty$ to replace $\mu^{c}_i(t)$ in~\eqref{eq:min_centroid_cluster} for the remaining similarity computations for the specific link between devices $i$ and $j$. 
In this way, each cluster $c \in \mathcal{C}_i$ links to at most one cluster $b \in \mathcal{C}_j$. 
% Specifically, we first need to find the cluster $c \in C_i$ with the smallest centroid difference to each cluster $b \in C_j$. 
% Formally, 
% \begin{equation} \label{eq:min_centroid_cluster}
%     c = \argmin_{c \in C_i} \vert \sigma_{i,j}^{c,b}(t) \vert,
% \end{equation}
% where $\sigma_{i,j}^{c,b}(t) = \mu_i^c(t) - \mu_j^b(t)$ and $\sigma_{k,i}^{b,c}(t)=0$ implying that clusters $b$ and $c$ are essentially identical. 
Then the data similarity $\lambda_{i,j}(t)$ between devices $i$ and $j$ can be computed as
\begin{equation} \label{eq:init_lambda_calc}
    \lambda_{i,j}(t) = \sum_{b \in \mathcal{C}_j} \min_{c \in \mathcal{C}_i} \frac{\vert \sigma^{c,b}_{i,j}(t) \vert} {\vert \mathcal{C}_i \vert}.
\end{equation}
For $t>0$, $\lambda_{i,j}(t)$ depends on the changes in $\sigma_{i,j}^{c,b}(t)$, which in turn depends on the data offloading ratio $\Phi_{i,j}^{c,b}(t-1)$ from the previous time step. %  $\Phi_{i,j}^{c,b}(t-1)$ from cluster $c$ at device $i$ to cluster $b$ at device $j$. 
% As a result of the more specific data offloading ratios, we can compute the aggregate data offloading ratio from $i$ to $j$ as 
Using these cluster-specific offloading ratios, we can compute the aggregate data offloading ratio from device $i$ to $j$ as 
\begin{equation} \label{eq:data_off_ratio_calc}
    \Phi_{i,j}(t) = \sum_{b \in \mathcal{C}_j} \sum_{c \in \mathcal{C}_i} \Phi_{i,j}^{c,b}(t). %\frac{\Phi_{i,j}^{c,b}(t) D_{i}^{c}(t)} {D_i(t)}. 
\end{equation}
Because devices lack exact information about each other's datasets, the selection of specific data for offloading from device $i$ to device $j$ also depends on the data similarity $\lambda_{i,j}(t)$ between them, and we explain the exact mechanics during the data offloading process in Sec.~\ref{ss:ovr_problem}. 
To capture both the node connectivity and data similarity jointly for offloading, we also define the connectivity-similarity matrix $\boldsymbol{\Lambda}(t) = [\Lambda_{i,j}(t)]$, and $\boldsymbol{\Lambda}(t) \triangleq \boldsymbol{\lambda}(t) \circ \mathbf{A}(t)$, where $\circ$ represents the Hadamard product.

{\color{black} This above process to determine the data similarities involves device-to-server transmissions, as devices will transmit their calculated centroids to the server for further processing such as cluster matching. These computations are performed prior to the start of our optimization process, outlined in Sec.~\ref{ss:ovr_problem}, and involve much smaller communication overhead than ML model transmission during the global aggregations. }

\vspace{-1mm}
\subsection{Distributed Machine Learning Model} \label{ss:mlp}
The learning objective of FedL is to train a global ML model, parameterized by a vector $\mathbf{w} \in \mathbb{R}^p$ (e.g., $p$ weights in a neural network), by training at edge devices. %using the devices participating in training. 
Formally, for $t \in \{0, \cdots, T\}$, each sampled device $i \in \mathcal{S}$ is concerned with updating its model parameter vector $\mathbf{w}_i(t)$ on its local dataset $\mathcal{D}_i(t)$. The local loss at device $i \in \mathcal{N}$ is defined as
\begin{align}
    F(\mathbf{w}_i(t)|\mathcal{D}_i(t)) = \frac {\sum_{d \in \mathcal{D}_i(t)} f(\mathbf{w}_i(t),\mathbf{x}_d,y_d)}{D_i(t)},
\end{align}
where $f(\mathbf{w}_i(t),x_d,y_d)$ denotes the corresponding loss (e.g., squared error for a regression problem, cross-entropy for a classification problem~\cite{wang2019adaptive}) of each datapoint $d \in \mathcal{D}_i$. Each device minimizes its local loss sequentially via gradient descent:
\begin{equation} \label{eq:gd}
    \mathbf{w}_i(t) = \mathbf{w}_i(t-1) - \eta \nabla F(\mathbf{w}_i(t-1)|\mathcal{D}_i(t)),
\end{equation}
where $\eta > 0$ is the step size and $\nabla F(\mathbf{w}_i(t-1)|\mathcal{D}_i(t))$ is the average gradient over the local dataset $\mathcal{D}_i(t)$. With periodicity $\tau$, the server performs a weighted average of $\mathbf{w}_i (k\tau)$, $i \in \mathcal{S}$:
%gathers the local model parameters of the sampled devices and aggregates them via a weighted average:
\begin{equation}\label{eq:w_agg}
    \mathbf{w}_{\mathcal{S}}(k\tau) = \frac{\sum_{i \in \mathcal{S}}  \Delta_i(k\tau) \mathbf{w}_i(k\tau)}{\sum_{i \in \mathcal{S}} \Delta_i(k\tau)},
\end{equation} 
where $k$ denotes the $k$-th aggregation, $k \in \{1,\cdots,K \}$, $K = \lceil T/\tau \rceil$, and $\Delta_{i}(k\tau) \triangleq  \sum_{t=(k-1)\tau+1}^{k\tau}D_i(t)$ denotes the total data located at node $i$ between $k-1$ and $k$. The server then synchronizes the sampled devices: $\mathbf{w}_i(k\tau) \leftarrow \mathbf{w}_\mathcal{S}(k\tau)$, $\forall i \in \mathcal{S}$. %, leading to the next round of local updates according to \eqref{eq:gd}.

Since we are concerned with the performance of the global parameter $\mathbf{w}_{\mathcal{S}}$, we  define $\mathbf{w}_{\mathcal{S}}(t)$ as the weighted average of $\mathbf{w}_i(t)$ as in~\eqref{eq:w_agg} for each time $t$, though it is only computed at the server when $t = k\tau$.
The global loss that we seek to optimize considers the loss over all the datapoints in the network:  %$\mathcal{D}_{\mathcal{N}}(t) = \cup_{i \in \mathcal{N}} \mathcal{D}_i(t)$ 
\begin{align} \label{eq:Lws}
\hspace{-4mm}
F\hspace{-0.5mm}\left(\mathbf{w}_{\mathcal{S}}(t)|\mathcal{D}_{\mathcal{N}}(t) \right)
= \frac{\sum_{i \in \mathcal{N} } D_i(t) F(\mathbf{w}_{\mathcal{S}}(t) \vert \mathcal{D}_i(t))}{D_{\mathcal{N}}(t)},
\hspace{-2mm}
\end{align}
where $\mathcal{D}_{\mathcal{N}}(t)$ denotes the multiset of the datasets of all the devices at time $t$, and $D_{\mathcal{N}}(t) \triangleq \sum_{i \in \mathcal{N}} D_i(t)$. 
%Since is impossible to determine the quantity of data overlaps between devices, we use $\sum_{i \in \mathcal{N}}D_i(t)$ as an approximation of the total unique data in order to perform the weighted average in \eqref{eq:Lws}. Although D2D is conducted, the union of the set of devices datasets does not change over time. 
% Let $\mathcal{D}_{\mathcal{N}}(t)$ denote the aggregate dataset: $\mathcal{D}_{\mathcal{N}}(t) \triangleq \cup_{i \in \mathcal{N}} \mathcal{D}_i (t)$, and the corresponding number of datapoints be $D_{\mathcal{N}}(t) \triangleq \sum_{i \in \mathcal{N}}D_i(t)$, $\forall t \in \{0,\cdots,T \}$. Throughout, for more compact notations we refer to $F\left(\mathbf{w}_{\mathcal{S}}(t)|\cup_{i \in \mathcal{N}}\mathcal{D}_{i}(t) \right)$ as $F\left(\mathbf{w}_{\mathcal{S}}(t)|\mathcal{D}_{\mathcal{N}}(t) \right)$.

% \subsection{Device Processing and Communication Models}
% \label{ss:device_models}
% Sampled device models here
% [device processing models]
% [device-to-server communication models here]
% Unsampled device models here
% [device-to-device data communication models here]

\subsection{Joint Sampling and Offloading Optimization} \label{ss:ovr_problem}
% Our scheme selects a set of devices $\mathcal{S}$ who then engage in performing local updates, communication to the server and data reception from the rest of the nodes.
The node selection procedure should consider device characteristics mentioned in Sec.~\ref{sss:devices} and the data similarity among the nodes represented by the connectivity-similarity matrix ${\boldsymbol{\Lambda}}(t)$. 
{\color{black}Thus, the goal of our optimization is to select (i) the subset of devices $\mathcal{S}^{\star}$ to sample from a total budget of $S$ and (ii) the offloading ratios $\Phi^{\star}_{i,j}(t)$ among the devices to jointly minimize the loss associated with $\mathbf{w}_\mathcal{S}(t)$ and the associated energy resource consumption (in terms of data processing energy and data offloading energy use).} We consider a time average for the objective, as devices may rely on intermediate aggregations for real-time inferences. For the variables, we define the binary vector $\mathbf{x} \triangleq (x_1,\cdots,x_N)$ to represent device sampling status, i.e., if $i \in \mathcal{S}$ then $x_i = 1$, otherwise $x_i=0$, and matrix $\boldsymbol{\Phi}(t) \triangleq [\Phi_{i,j}(t)]_{1 \leq i,j \leq N}$ to represent the offloading ratios at time $t$. The resulting optimization problem $\boldsymbol{\mathcal{P}}$ is as follows:
\begin{align}
& (\boldsymbol{\mathcal{P}}):~
\underset{\mathbf{x},\{\boldsymbol{\Phi}(t) \}_{t=1}^{T}} 
{\textrm{minimize}}~
\frac{1}{T}\sum_{t=1}^{T}  \bigg( \underbrace{ \alpha F(\mathbf{w}_\mathcal{S}(t)|\mathcal{D}_{\mathcal{N}}(t))}_{(a)}  \label{eq:obj_new} \\ 
& + \underbrace{\beta \sum_{i \in \mathcal{S}} {E}_{i}^{P}(t)}_{(b)} + \underbrace{\gamma \sum_{k \in \mathcal{N}} {E}_{k}^{Tx}(t)}_{(c)} \bigg) \nonumber \\ 
&\textrm{subject to}\nonumber\\ 
%% constraints begin here
& D_i^c(t) = D_i^c(t-1) + R_{i}^c(t), ~i\in\mathcal{N}, \label{eq:con1}\\
& R_{i}^{c}(t) = \sum_{k \in \mathcal{N}} 
\sum_{b \in \mathcal{C}_k} 
D_k^{b}(t-1)\Phi_{k,i}^{b,c}(t){\Lambda}_{k,i}(t-1), i\in\mathcal{N}, \hspace{-1mm} \label{eq:con2}\\ 
% &\Lambda_{k,i}(t) \hspace{-1mm} = \hspace{-1mm} \Lambda_{k,i}(t \hspace{-0.6mm} - \hspace{-0.6mm} 1) 
% \hspace{-1mm}
% + \hspace{-1mm} (1 \hspace{-0.6mm} - \hspace{-0.6mm} \Lambda_{k,i}(t \hspace{-0.5mm} - \hspace{-0.5mm} 1))\frac{\Phi_{k,i}(t)}{C_k}, i,k\in\mathcal{N}, \label{eq:con3} \\
% &\sum_{i \in \mathcal{N}} \frac{\Phi_{k,i}(t)}{C_k} \leq 1, ~k\in\mathcal{N}, \\ %%% double check this guy 
&(1-x_i)(1-x_k)\Phi_{k,i}(t) = 0, ~i,k \in \mathcal{N}, \label{eq:con4} \\ 
&x_k\Phi_{k,i}(t) = 0, ~i,k\in \mathcal{N}, \label{eq:con5} \\ 
&(1-A_{k,i}(t)) \Phi_{k,i}(t) = 0, ~i,k \in \mathcal{N}, \label{eq:con6} \\
&\sum_{i\in \mathcal{N}} x_i=S, \label{eq:con7}\\
%% entirely new constraints
& \sum_{c \in \mathcal{C}_i} D_{i}^{c}(t) = D_i(t), ~i \in \mathcal{N}, \label{eq:con8} \\
& \sum_{c \in \mathcal{C}_i} \sum_{b \in \mathcal{C}_k} \Phi_{k,i}^{b,c} (t) = \Phi_{k,i}(t), ~i,k \in \mathcal{N}, \label{eq:con9} \\ 
% & \sum_{c \in \mathcal{C}_i} \Phi_{k,i}^{b,c}(t) \leq 1\\
& \sum_{b \in \mathcal{C}_k} \Phi_{k,i}^{b,c}(t) \leq C_k, ~i,k \in \mathcal{N}, \label{eq:con10} \\ 
& \sigma_{k,i}^{b,c}(t) \hspace{-1mm} = \hspace{-1mm} \sigma_{k,i}^{b,c}(t \hspace{-0.5mm} - \hspace{-0.5mm} 1)(1 \hspace{-0.5mm} - \hspace{-0.5mm} \Phi_{k,i}^{b,c}(t)), b \in \mathcal{C}_k,  c \in \mathcal{C}_i, i,k \in \mathcal{N}, \hspace{-2mm} \label{eq:con11} \\
% & \sigma_{k,i}^{b,c}(0) = \mu_k^{b} - \mu_i^{c} \\
& \sum_{\substack{b \in \mathcal{C}_k \\ 
b \neq \argmin \vert \sigma^{b,c}_{k,i}(t) \vert}} 
\Phi_{k,i}^{b,c}(t) = 0,
~c \in \mathcal{C}_i, ~i,k \in \mathcal{N}, \label{eq:con12}\\
& E_i^{P}(t) = p_i(t) D_i(t) \leq P_i(t), i \in \mathcal{N}, \label{eq:con13} \\ 
& E_k^{Tx}(t) = D_k(t) \sum_{i \in \mathcal{S}} {\Phi_{k,i}(t) \psi_{k,i}(t)} \leq \Psi_{k}(t), k \in \mathcal{N}. \label{eq:con14} \\
& 0 \leq \Lambda_{k,i}(t) \leq 1, x_i \in \{0,1\}, ~i,k\in\mathcal{N}, \label{eq:con15} \\ 
& 0 \leq \Phi_{k,i}^{b,c}(t) \leq 1,  b \in \mathcal{C}_k,  c \in \mathcal{C}_i, ~i,k \in \mathcal{N}. \label{eq:con16} 
\end{align} 
% \end{equation}

\noindent 
The objective function~\eqref{eq:obj_new} captures the balance between three key terms: the FedL loss in~\eqref{eq:obj_new}(a) scaled by $\alpha$, the data processing resource use in~\eqref{eq:obj_new}(b) scaled by $\beta$, and the D2D offloading resource use in~\eqref{eq:obj_new}(c) scaled by $\gamma$. 
The data at sampled devices, i.e., $D_i(t)$ for $i \in \mathcal{S}$, changes over time in~\eqref{eq:con1} based on the total received data $R_i(t)$ for device $i$. $R_i(t)$ is determined in~\eqref{eq:con2} by scaling the data transmissions from each device $k \in \hat{\mathcal{S}}$ to device $i$ according to the data similarity, which we explain further in the next paragraph. 
% In response to the data offloading, the connectivity-similarity matrix is updated in~\eqref{eq:con3}. 
% Together,~\eqref{eq:con2} and~\eqref{eq:con3} capture similarity-aware D2D offloading, which we explain further in the following paragraph. 
Through~\eqref{eq:con4}-\eqref{eq:con6}, offloading only occurs between single-hop D2D neighbors from $k \in \hat{\mathcal{S}}$ to $i \in \mathcal{S}$, while~\eqref{eq:con7} ensures that $(\boldsymbol{\mathcal{P}})$ maintains compliance with the desired sampling size, i.e., $|\mathcal{S}^{\star}| = S$.
The two expressions~\eqref{eq:con8}-\eqref{eq:con9} capture the definitions for $D_i(t)$ and $\Phi_{k,i}(t)$. 
Next,~\eqref{eq:con10} ensures that the total offload rate at each device cannot exceed the quantity of its local data clusters,~\eqref{eq:con11} is the update rule for centroid differences between two clusters $b$ and $c$ at two difference devices $i$ and $k$, and~\eqref{eq:con12} clarifies that data offloading only occurs among the matched data clusters as described in Sec.~\ref{sss:graph}. 
The constraints~\eqref{eq:con13}-\eqref{eq:con14} ensure that our D2D offloading solution adheres to maximum device data processing limits $P_i(t)$, and D2D communication limits $\Psi_i(t)$. Finally,~\eqref{eq:con15}-\eqref{eq:con16} express the lower and upper limits of the connectivity-similarity values and the data offloading variables.

\textbf{Similarity-aware D2D offloading:} The amount of raw data device $i$ receives from $k$ is $\sum_{b \in \mathcal{C}_k} \Phi_{k,i}^{b,c}(t) D_{k}^{b}(t-1)$. %, which does not consider the connectivity-similarity factor $\Lambda_{k,i}(t-1)$. 
Ideally, device $i$ will receive data from its neighbors that is dissimilar to $\mathcal{D}_i(0)$. 
However, neither $i$ nor $k$ have full knowledge of each others' datasets in this distributed scenario (nor does the server). 
% (i.e., i.i.d.) 
Therefore, data offloading in $\boldsymbol{\mathcal{P}}$ is conducted through a uniformly random selection of $\Phi_{k,i}^{b,c}(t) D_{k}^{b}(t-1)$ data points (for each data cluster $b \in \mathcal{C}_k$) from $k$ to send to $i$. 
From the viewpoint of cluster $b$ at device $k$, the estimated overlapping data that arrives at $i$ is $\Phi_{k,i}^{b,c}(t) D_{k}^{b}(t-1) (1-\Lambda_{k,i}(t-1))$, and the resulting useful data is $\Phi_{k,i}^{b,c}(t) D_{k}^{b}(t-1)\Lambda_{k,i}(t-1)$. 
%Additionally, it is impossible to obtain the exact similarity between the nodes post-offloading as this requires exact knowledge of the datasets of the nodes at every time instance. % (this information does not exist anywhere, even at the server)
We use the entirety of $\Lambda_{k,i}(t)$ for this estimation because it is possible for data in cluster $b$ to have overlap with multiple clusters in $\mathcal{C}_i$ at device $i$. So, overall data similarity in the form of $\Lambda_{k,i}(t)$ provides a better holistic estimate that an incoming data from device $k$ to $i$ will be similar to \textit{any} existing data at device $i$. 
% the overall data similarity from all clusters $\mathcal{C}_i$ at device $i$ provides a better holistic estimate of the likelihood that an incoming data arrival will be similar to \textit{any} existing data at device $i$.
% the data at the fringes of each data cluster, and so using the overall data similarity from all clusters $\mathcal{C}_i$ at device $i$ provides a better holistic estimate of the likelihood that an incoming data arrival will be similar to \textit{any} existing data at device $i$.} 
%%%% especially, because you don't know which data points are going to be more similar and which aren't when you offload data from device k to device i
%and, furthermore, not all data arrivals to a cluster $c$ are necessarily optimal at cluster $c$
Subsequently, we update $\sigma_{k,i}^{b,c}(t)$ per~\eqref{eq:con11} and then use the data similarity calculation in~\eqref{eq:init_lambda_calc} to update $\Lambda_{k,i}(t)$.
In particular, when $k$ transfers all of its data to $i$ (i.e., when $\Phi_{k,i}^{b,c}(t)=1$ $\forall b \in \mathcal{C}_k$), 
$\sigma_{k,i}^{b,c} = 0$ $\forall b \in \mathcal{C}_k$ and $\forall c \in \mathcal{C}_i$, leading to $\Lambda_{k,i}(t) \rightarrow 0$, thus preventing further data offloading  according to~\eqref{eq:con2}. 
Imposing these constraints promotes data diversity among the sampled nodes through offloading.

\textbf{Solution overview:} Problem~$\boldsymbol{\mathcal{P}}$ faces two major challenges: (i) it requires a concrete model of the loss function in~\eqref{eq:obj_new}(a) with respect to the datasets, which is, in general, intractable for deep learning models~\cite{goodfellow2016deep}, and (ii) even if the loss function is known, the coupled sampling and offloading procedures make this problem an NP-hard mixed integer programming problem. To overcome this, we will first consider the offloading subproblem for a fixed sampling strategy, and develop a sequential convex programming method to solve it in Sec.~\ref{s:p1}. Then, we will integrate this into a graph convolutional network (GCN)-based methodology that learns the relationship between the network properties, sampling strategy (with its corresponding offloading), and the resulting FedL model accuracy in Sec.~\ref{s:p2}. An overall flowchart of our methodology is given in Fig.~\ref{fig:BigPicture}.

%In the following section, we first assume a known set of sampled nodes, i.e., known $\mathbf{x}^*$ and $\mathcal{S}$, and propose an approximate format for the loss function that prioritizes the number of training data points, based on deriving an upper bound on the performance of FedL.
\vspace{-1mm}
\section{Developing the Offloading Optimizer} %D2D Data Offloading Optimization}
\label{s:p1}
\noindent In this section, we study the offloading optimization subproblem of $\boldsymbol{\mathcal{P}}$. Specifically, our theoretical analysis of~\eqref{eq:obj_new}(a) under common assumptions will yield an efficient approximation of the FedL loss objective in terms of the offloading variables (Sec.~\ref{ss:p1b}). We will then develop our solver for the resulting optimization (Sec.~\ref{ss:p1c}). %~\eqref{eq:obj}-\eqref{eq:con10}

\vspace{-1mm}
\subsection{Definitions and Assumptions} \label{ss:p1def}
To aid our theoretical analysis of FedL, similar to \cite{wang2019adaptive}, we will consider a hypothetical ML training process that has access to the entire dataset $\mathcal{D}_{\mathcal{N}}(t)$ at each time instance. The parameter vector $\mathbf{v}_k(t)$ for this centralized model is trained as follows: (i) at each global aggregation $t = k\tau$, $\mathbf{v}_k(t)$ is synchronized with $\mathbf{w}_{\mathcal{S}}(t)$, i.e., $\mathbf{v}_k(t) \leftarrow \mathbf{w}_{\mathcal{S}}(t)$, and (ii) in-between global aggregation periods, $\mathbf{v}_k(t)$ is trained based on gradient descent iterations to minimize the global loss $F(\mathbf{v}_k(t) \vert \mathcal{D}_{\mathcal{N}}(t))$. %We will obtain an upper bound on the difference between FedL model parameters with D2D data offloading versus classic centralized learning.%D2D data offloading underbitrary sampling schemes. In our process to get this upper bound, inspired by~\cite{wang2019adaptive}, we consider an auxiliary set of parameters $\mathbf{v}_k(t)$, using $F(\mathbf{v}_k(t)|\mathcal{D}_{\mathcal{N}}(t))$ as the corresponding model loss, to denote the ML model parameters in classic centralized learning, which is trained based on the entire dataset $\mathcal{D}_{\mathcal{N}}(t)$ being hypothetically stored at the server at each time instance. We utilize the performance gap between the sampled FedL and the classic centralized learning within global aggregations and resynchronize $\mathbf{v}_k$ to $\mathbf{w}_{\mathcal{S}}$ at global aggregations time instances $k\tau$, i.e., $\mathbf{v}_k(k\tau) \leftarrow \mathbf{w}_{\mathcal{S}}(k\tau)$, $k \in \{1,\cdots,K\}$. We then obtain an upper bound on the difference between FedL model parameters under arbitrary device sampling with D2D data offloading versus classic centralized learning. 
%In the following, we define and obtain the difference of gradients between the sampled FedL with D2D offloading and centralized learning. 
%We focus on the difference between the evolution of gradients of the aforementioned  Fws and Fvk, $\nabla F$ The aggregated sampled parameters $\mathbf{w}_\mathcal{S}$ are trained on gradient descent using $\nabla F(\mathbf{w}_{\mathcal{S}}(t)\vert \mathcal{D}_{\mathcal{S}}(t))$. 
%However, due to real-time offloading, it is difficult to estimate $\mathcal{D}_{\mathcal{S}}(t)$ for all $t \in \{1,\cdots,T \}$. We instead use the total empirical gradient of the sampled parameters $\nabla F(\mathbf{w}_{\mathcal{S}}(t)|\mathcal{D}_{\mathcal{N}})$ on sampled FedL networks. We derive the approximation error from using $\nabla F(\mathbf{w}_{\mathcal{S}}(t)|\mathcal{D}_\mathcal{N})$ in the following lemma.
%We define the error in the gradient $F(\mathbf{w}_{\mathcal{S}}(t)|\mathcal{D}_{\mathcal{N}})$ due to sampling as:
%We will also find it useful to define the following quantities on the differences between gradients: %
\begin{definition} [Difference between sampled and unsampled gradients] \label{def:e} %\label{eq:l1_result}
We define the instantaneous difference between $\nabla F(\mathbf{w}_{\mathcal{S}}(t)\vert \mathcal{D}_{\mathcal{N}}(t))$, the gradient with respect to the full dataset across the network, and $\nabla F(\mathbf{w}_{\mathcal{S}}(t)\vert \mathcal{D}_{\mathcal{S}}(t))$, the gradient with respect to the sampled dataset, as:
\begin{equation}
\hspace{-0.0mm}
\begin{aligned}  \label{eq:l1_result} %\label{eq:l1_1} 
&\zeta(\mathbf{w}_{\mathcal{S}}(t)) \triangleq \frac{G_{\mathcal{S}}(t)}{D_{\mathcal{S}}(t)} - \frac{\sum_{i \in \mathcal{N}} D_i(t) \nabla F(\mathbf{w}_{\mathcal{S}}(t)|\mathcal{D}_i(t))}{D_{\mathcal{N}}(t)},
\end{aligned}
\hspace{-4mm}
\end{equation} %$G_{\mathcal{S}}(t)$ is the scaled sum of gradients on the sampled datasets:$D_{\mathcal{S}}(t)$ is the sum of data at the unsampled devices:
where~$G_{\mathcal{S}}(t) \triangleq \sum_{i \in \mathcal{S}}D_i(t)\nabla F(\mathbf{w}_{\mathcal{S}}(t)|\mathcal{D}_i(t))$ is the scaled sum of gradients on the sampled datasets, and~$D_{\mathcal{S}}(t) = \sum_{i \in \mathcal{S}}D_i(t)$ is the total data across the sampled devices.
\end{definition}
% \begin{lemma} \label{lem:loss}
% The instantaneous difference of gradients in Definition~\ref{def:e} can be expressed as:

% \vspace{-0.1in}
% \small
% \begin{equation}
% \hspace{-3mm}
% \begin{aligned} \label{eq:l1_result}
% \zeta(\mathbf{w}_{\mathcal{S}}(t))\hspace{-.6mm}  =\hspace{-.6mm}  \frac{D_{\mathcal{N}}(t) \hspace{-.6mm} -\hspace{-.6mm}  D_{\mathcal{S}}(t)}{D_{\mathcal{N}}(t) D_{\mathcal{S}}(t)} G_{\mathcal{S}}(t)\hspace{-.6mm}  - \hspace{-.6mm} \nabla F(\mathbf{w}_{\mathcal{S}}(t)|\mathcal{D}_{\mathcal{N}}(t)) \hspace{-.6mm} + \hspace{-.6mm} \frac{G_{\mathcal{S}}(t)}{D_{\mathcal{N}}(t)}.
% \end{aligned} 
% \hspace{-2mm}
% \vspace{-2mm}
% \end{equation}
% \normalsize
% \end{lemma}

% \begin{proof}
% Adding and subtracting $\frac{\sum_{j \in \mathcal{S}} D_j(t) G_{\mathcal{S}}(t)}{D_{\mathcal{N}}(t) D_{\mathcal{S}}(t)}$ to \eqref{eq:l1_1} and rearranging leads to the result. %which results in 
% % \begin{equation}
% % \begin{aligned}
% % & \Xi(\mathbf{w}_{\mathcal{S}}(t)) = \frac{D_{\mathcal{N}} G_{\mathcal{S}}(t)- \sum_{j\in \mathcal{S}} D_j(t) G_{\mathcal{S}}(t)}{D_{\mathcal{N}} D_S(t)} \\
% % & - \frac{\sum_{i \in \mathcal{N}}D_i(t) \nabla L(\mathbf{w}_{\mathcal{S}}(t)|\mathcal{D}_i(t))}{D_{\mathcal{N}}} + \frac{\sum_{j \in \mathcal{S}}D_j(t)G_{\mathcal{S}}(t)}{D_{\mathcal{N}} D_{\mathcal{S}}(t)}.
% % \end{aligned}
% % \end{equation}
% % As $G_{\mathcal{S}}(t)$ is independent of $j$, rearranging then gives the result.
% \end{proof}
\begin{definition}[Difference between sampled and unsampled gradients] We define $\delta_i(t)$ as the upper bound between the gradient computed on $\mathcal{D}_i(t)$ for $i \in \mathcal{S}$ and $\mathcal{D}_{\mathcal{N}}(t)$ at time $t$: %in terms of $\zeta(\mathbf{w}_{\mathcal{S}}(t))$: 
\label{def:deltai}

% \vspace{-0.16in}
\vspace{-2mm}
\small
\begin{equation}
% \hspace{-.01mm}
\Vert \nabla F(\mathbf{w}_{\mathcal{S}}(t)|\mathcal{D}_i(t)) \hspace{-.45mm}- \hspace{-.45mm} \nabla F(\mathbf{w}_{\mathcal{S}}(t)|\mathcal{D}_{\mathcal{N}}(t))\hspace{-.45mm} -\hspace{-.45mm} \zeta(\mathbf{w}_{\mathcal{S}}(t)) \hspace{-.4mm}\Vert\hspace{-.5mm}\leq\hspace{-.5mm} \delta_i(t).
% \hspace{-14mm}
\end{equation}
\end{definition}

%Next, we define $\delta_i(t)$ as the upper bound between the gradient computed on $\mathcal{D}_i(t)$ and $\mathcal{D}_{\mathcal{N}}(t)$ at time $t$, given by:
% \small
% \begin{equation}
% \hspace{-3mm}
%      \Vert \nabla F(\mathbf{w}_{\mathcal{S}}(t)|\mathcal{D}_i(t)) \hspace{-.5mm}- \hspace{-.5mm} \nabla F(\mathbf{w}_{\mathcal{S}}(t)|\mathcal{D}_{\mathcal{N}}(t))\hspace{-.5mm} -\hspace{-.5mm} \Xi(\mathbf{w}_{\mathcal{S}}(t)) \Vert\hspace{-.5mm}\leq\hspace{-.5mm} \delta_i(t).
%      \hspace{-3mm}
% \end{equation} 
% \normalsize We assume the following properties for the global loss function: 

\vspace{-0.05in}
We also make the following standard assumptions~\cite{wang2019adaptive,wang2021network} on the loss function $F(\mathbf{w})$ for the ML model being trained: 
\begin{assumption} 
We assume $F(\mathbf{w})$ is convex with respect to $\mathbf{w}$, L-Lipschitz, i.e., $\Vert F(\mathbf{w}) - F(\mathbf{w}^{\prime}) \Vert \leq L\Vert \mathbf{w} - \mathbf{w}^{\prime}\Vert $, and  $\beta$-smooth, i.e., $\Vert \nabla F(\mathbf{w}) - \nabla F(\mathbf{w}^{\prime}) \Vert \leq \beta \Vert \mathbf{w} - \mathbf{w}^{\prime}\Vert$, $\forall \mathbf{w}, \mathbf{w}^{\prime}$.
\end{assumption}
Despite these assumptions, we will show in Sec.~\ref{s:numRes} that our results still obtain significant improvements in practice for neural networks which do not obey the above assumptions.
%our results still obtain significant improvements in practice for neural networks (which may not obey the above assumptions) in Sec.~\ref{s:numRes}.
%we will show in Sec.~\ref{s:numRes} that our results still obtain significant improvements in practice for neural networks (which are non-convex).
\vspace{-1mm}
\subsection{Upper Bound on Convergence} \label{ss:p1b}
For convergence analysis, we assume that devices only offload the same data once, and assume that recipient nodes always keep received data. This must be done to ensure that the optimal global model parameters remain constant throughout time.
The following theorem gives an upper bound on the difference between the parameters of sampled FedL and those from the centralized learning, i.e., $\Vert \mathbf{w}_{\mathcal{S}}(t) - \mathbf{v}_k(t) \Vert$, over time:
%$F(w)$ is convex with respect to $w$, $L$-Lipschitz: $\Vert F(\mathbf{w}) - F(\mathbf{w}^{\prime}) \Vert \leq L\Vert \mathbf{w} - \mathbf{w}^{\prime}\Vert $ $\forall \mathbf{w}, \mathbf{w}^{\prime}$, and $\beta$-smooth $\Vert \nabla F(\mathbf{w}) - \nabla F(\mathbf{w}^{\prime}) \Vert \leq \beta \Vert \mathbf{w} - \mathbf{w}^{\prime}\Vert$ $\forall \mathbf{w}, \mathbf{w}^{\prime}$. Using the result of Lemma \ref{lem:loss}, we find an upper bound on the difference between the parameters of sampled FedL and the parameters from centralized learning, i.e., $\Vert \mathbf{w}_{\mathcal{S}}(t) - \mathbf{v}_k(t) \Vert $, $\forall t$, in the following theorem.

\begin{theorem}[Upper bound on the difference between sampled FedL and centralized learning]\label{thm:error} %\label{the:main1}
Assuming $\eta \leq {\beta}^{-1}$, the upper-bound on the difference between $\mathbf{w}_{\mathcal{S}}(t)$ and $\mathbf{v}_k(t)$ within the local update period before the $k$-th global aggregation, $t \in \{(k-1)\tau+1,...,k\tau\}$, is given by:
\vspace{-4.5mm}

\small
\begin{equation} \label{th1:1}
\hspace{-0mm}
    \Vert \mathbf{w}_{\mathcal{S}}(t) - \mathbf{v}_{k}(t)\Vert \leq \frac{1}{\beta} \hspace{-5mm} \sum_{y = (k-1)\tau+1}^{t} \hspace{-1mm} \bigg(\Upsilon (y,k) + \Vert \zeta(\mathbf{w}_{\mathcal{S}}(y-1))\Vert\bigg),
    \hspace{-3mm}
\end{equation}
\normalsize
where $\Upsilon (y,k) \triangleq \delta_{\mathcal{S}}(y) (2^{y-1-(k-1)\tau}-1)$,
and
\begin{equation} \label{th1:3}
   \delta_{\mathcal{S}}(t) \triangleq \left({\sum_{i \in \mathcal{S}}D_i(t)\delta_i(t)}\right)\left({\sum_{i \in \mathcal{S}}D_i(t)}\right)^{-1}.
\end{equation} %,~t \in \{0\} \cup \mathbb{N}. \hspace{-1mm} 
\end{theorem}

\begin{proof}
See Appendix~\ref{app:main1}.
\end{proof}
Through $\delta_{\mathcal{S}}(t)$, Theorem~\ref{thm:error} establishes a relationship between the difference in model parameters and the datapoints $D_i(t)$ in the sampled set $i \in \mathcal{S}$. Using this, we obtain an upper bound on the difference between our $\mathbf{w}_{\mathcal{S}}(t)$ and the global minimizer of model loss $\mathbf{w}^*(t) =\argmin_{\mathbf{w}}   F(\mathbf{w}|\mathcal{D}_{\mathcal{N}}(t))$:

\begin{corollary}[Upper bound on the difference between sampled FedL and the optimal] \label{c1}
The difference of the loss induced by $\mathbf{w}_{\mathcal{S}}(t)$ compared to the loss induced by $\mathbf{w}^*(t)$ for $t \in \{(k-1)\tau,\cdots,k\tau-1\}$, is given by:
\begin{equation} \label{eq:cl1_result}
\begin{aligned}
& F(\mathbf{w}_{\mathcal{S}}(t)|\mathcal{D}_{\mathcal{N}}(t)) - F(\mathbf{w}^*(t)|\mathcal{D}_{\mathcal{N}}(t))\leq  \\
&g(\hat{\Upsilon}(\hat{K})) \triangleq \left(t \xi \eta \left(1 - \frac{\beta \eta}{2}\right) - \frac{(\hat{K}+1)L}{\beta \epsilon^2} \hat{\Upsilon}(\hat{K}) \right)^{-1} ,
\end{aligned}
\end{equation}
\normalsize
where $\hat{\Upsilon}(\hat{K})\hspace{-0.5mm} \triangleq\hspace{-0.5mm} \sum_{y=(\hat{K}-1)\tau+1}^{\hat{K}\tau}  \left(\Upsilon(\hat{K},y) \hspace{-0.5mm}+\hspace{-0.5mm} \Vert \zeta(\mathbf{w}_{\mathcal{S}}(y\hspace{-0.5mm}-\hspace{-0.5mm}1))\Vert \right)$, $\hat{K} = \floor{t /\tau}$, and $\xi = \min_k \frac{1}{{\Vert \mathbf{v}_k((k-1)\tau) - \mathbf{w}^*(t) \Vert}^2}$.
\end{corollary}

\begin{proof}
See Appendix~\ref{app:c1}.
\end{proof}
As our ultimate goal is an expression of~\eqref{eq:obj_new}(a) in terms of the data $\mathcal{D}_i(t)$ at each node, we consider the relationship between $g(\hat{\Upsilon}({\hat{K}}))$ and $\mathcal{D}_i(t)$, which is clearly non-convex through $\hat{\Upsilon}({\hat{K}})$. Since $\hat{\Upsilon}(\hat{K}) \ll 1$ (see Appendix~\ref{app:c1}),~\eqref{eq:cl1_result} can be approximated using the first two terms of its Taylor series:
%The relationships on the right hand side (RHS) of \eqref{eq:cl1_result} are very complicated. We interpret it as a linear function $g(\hat{\Upsilon})$ using the first two terms of its Taylor expansion, obtaining:
\begin{equation} \label{eq:fupsilon}
\begin{aligned}
& g(\hat{\Upsilon}) \approx \frac{1}{t \xi \eta (1 -\frac{\eta \beta}{2})} + \frac{{(\hat{K}+1)L}}{\beta \epsilon^2\left(t \xi \eta (1 -\frac{\eta \beta}{2})\right)^2} \hat{\Upsilon}.
\end{aligned}
\end{equation} 
At each time instant, the first term in the right hand side (RHS) of \eqref{eq:fupsilon} is a constant. Thus, under this approximation, the RHS of \eqref{eq:cl1_result} becomes proportional to $\hat{\Upsilon}$, which is in turn a function of $\delta_{\mathcal{S}}(t)$. The final step is to bound the expression for $\delta_i(t)$, and thus their weighted sum $\delta_{\mathcal{S}}(t)$, in terms of the $D_i(t)$, $\forall i \in \mathcal{S}$.

%Using the result in \eqref{eq:fupsilon} and ignoring the constant term, the RHS of \eqref{eq:cl1_result} becomes proportional to $\hat{\Upsilon}$, which is in turn a function of $\delta_{\mathcal{S}}(t)$. The following proposition bounds the expressions for $\delta_i(t)$, and thus their weighted sum $\delta_{\mathcal{S}}(t)$.

%Treating \eqref{eq:fupsilon} as a function of $\hat{\Upsilon}$ and using this result in \eqref{eq:cl1_result}, we get that the RHS of \eqref{eq:cl1_result} becomes a function of $\hat{\Upsilon}$, 

\begin{proposition}[Upper bound on the difference between local gradients] \label{prop:1} The difference in gradient with respect to a sampled device dataset vs. the full dataset satisfies:
\begin{equation}\label{eq:lemma}
\hspace{-0mm}
\begin{aligned} 
&\Vert\nabla F(\mathbf{w}_{\mathcal{S}}(t)|\mathcal{D}_i(t)) - \nabla F\left(\mathbf{w}_{\mathcal{S}}(t)|\mathcal{D}_{\mathcal{N}}(t) \right) - \zeta(\mathbf{w}_{\mathcal{S}}(t))\Vert \\
&  \leq \left(\frac{D_{\mathcal{N}}(t)-D_{\mathcal{S}}(t)}{D_{\mathcal{N}}(t)}\right) \overline{\nabla F(t)} + \frac{\gamma}{\sqrt{D_i(t)}} + C \equiv \delta_i(t), 
\end{aligned}
\hspace{-5mm}
\end{equation}
where $C \triangleq \left({D_{\mathcal{N}}(t)}\right)^{-1}\sum_{i \in \hat{\mathcal{S}}} D_i(t)\nabla F(\mathbf{w}_{\mathcal{S}}(t)|\mathcal{D}_i(t))$, $\hat{\mathcal{S}} = \mathcal{N} \setminus \mathcal{S}$, $\gamma$ is a constant independent of $\mathcal{D}_i(t)$, and 
\begin{equation} \label{eq:nablaFbar} %and $\overline{\nabla F(t)}$, 
\overline{\nabla F(t)} \triangleq \left({D_{\mathcal{S}}(t)}\right)^{-1}{\sum_{i \in \mathcal{S}}D_i(t)\nabla F(\mathbf{w}_{\mathcal{S}}(t)|\mathcal{D}_i(t))}.
\end{equation}
%$\overline{\nabla F(t)} \triangleq \left({D_{\mathcal{S}}(t)}\right)^{-1}{\sum_{i \in \mathcal{S}}D_i(t)\nabla F(\mathbf{w}_{\mathcal{S}}(t)|\mathcal{D}_i(t))}$.
% $C$ is independent of all offloading as devices $i \in \hat{\mathcal{S}}$ do not receive data, and therefore is represented as a constant.
\end{proposition}
\begin{proof}
See Appendix~\ref{app:prop1}.
\end{proof}
\normalsize
The above proposition relates each $\delta_i(t)$ to the number of instantaneous data points available at device $i$. %, which we will use next to develop our sequential convex optimization method.

\begin{figure*}[t]
\includegraphics[width=.78\textwidth]{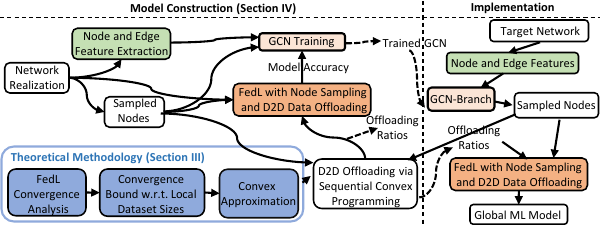} 
\centering
% \caption{Overview of the joint sampling and offloading methodology developed in Sec.~\ref{s:p1}\&\ref{s:p2}. During model construction, the offloading optimizer from Sec.~\ref{s:p1} is used to determine offloading for a set of sampled devices. The GCN-based algorithm developed in Sec.~\ref{s:p2} determines the combination of sampling and optimized offloading expected to maximize FedL accuracy. Then, the resulting model is applied to the target network for FedL implementation.}
\caption{\color{black}Overview of the joint sampling and offloading methodology developed in Sec.~\ref{s:p1}\&\ref{s:p2}. During model construction, our methodology trains a GCN, using various network realizations and sampled sets of nodes with the data offloading optimization from Sec.~\ref{s:p1}. 
At the implementation stage, the target network uses the GCN-based algorithm developed in Sec.~\ref{s:p2} to obtain a sampled set of devices, which then undergo the D2D data offloading optimization process. 
Finally, we apply the results of the sampling and D2D data offloading processes for FedL, yielding a global ML model after training completion.}
\label{fig:BigPicture}
\vspace{-1mm}
\end{figure*}

\subsection{Offloading as a Sequential Convex Optimization} \label{ss:p1c}
%The first term of (29) is not a function of the number of datapoints, therefore we focus on the second term.

Using the result of~\eqref{eq:fupsilon} to replace the RHS of~\eqref{eq:cl1_result} implies that the FedL loss term in the objective function (i.e.,~\eqref{eq:obj_new}(a)) is proportional to $\frac{1}{T}\sum_{t=1}^{T} \delta_{\mathcal{S}}(t)$, where $\delta_{\mathcal{S}}(t)$ is defined in~\eqref{th1:3} as a sum-of-ratios of $\delta_i(t)$. Considering $\frac{1}{T}\sum_{t=1}^{T} \delta_{\mathcal{S}}(t)$ as the objective in problem $\boldsymbol{\mathcal{P}}$ yields the sum-of-ratios problem in fractional programming~\cite{schaible2003fractional}. %an NP-complete problem and one of the hardest open problems in optimization \cite{schaible2003fractional}.
The scale of existing solvers for the sum-of-ratios fractional programming problem (e.g. \cite{kuno2002branch}) %, \cite{konno2000branch}
are on the order of ten ratios, which corresponds to ten devices in our case. Contemporary large-scale networks that may have hundreds of edge devices~\cite{8373692} therefore cannot be solved accurately or in a time-sensitive manner. 
Motivated by the above fact, we approximate $\delta_{\mathcal{S}}(t) \approx \frac{1}{S} \sum_{i \in \mathcal{S}} \delta_i(t)$. 
Using this with~\eqref{eq:lemma},
{\color{black} we obtain the following approximation for the loss function term in~\eqref{eq:obj_new}(a): 
\vspace{-4mm}

\small
\begin{align} \label{eq:obj_temp}
\widetilde{F}(t) = \bigg[ \underbrace{\left(\frac{D_{\mathcal{N}}(t)-D_{\mathcal{S}}(t)}{D_{\mathcal{N}}(t)}\right)  \overline{\nabla F(t)}}_{(a)} + \frac{1}{\vert\mathcal{S}\vert} \sum_{i \in \mathcal{S}}\underbrace{ \frac{\gamma}{\sqrt{D_i(t)}}}_{(b)} \bigg],
\end{align}
}
\vspace{-3mm}

%Large quantities of devices, such as those on large-scale networks, therefore cannot be solved accurately or in a time-sensitive manner. Motivated by the above fact, we approximate $\delta_{\mathcal{S}}(t) \approx \frac{1}{S} \sum_{i \in \mathcal{S}} \delta_i(t)$, where the approximation gets more accurate as the sampled nodes have similar amount of data. Finally, we approximate the objective function of~\eqref{eq:obj} via considering the terms that explicitly consider the sampled datapoints in \eqref{eq:lemma} as:
%We obtain the following substitute expression for $\delta_{\mathcal{S}}(t)$ get that $F(\mathbf{w}_{\mathcal{S}}(t)|\mathcal{D}_{\mathcal{N}})-F(\mathbf{w}^*|\mathcal{D}_{\mathcal{N}}) \propto$
%\begin{align}
% $\sum_{i \in \mathcal{S}}\left(\frac{D_{\mathcal{N}}-D_{\mathcal{S}}(t)}{D_{\mathcal{N}}}\right) \overline{\nabla F(t)}+ \frac{\gamma}{\sqrt{D_i(t)}}$, where
%\end{align}
% \vspace{-4mm}

\normalsize
% \vspace{-3mm}
where term $(a)$ is due to sampling and term $(b)$ is the statistical error from the central limit theorem. Thus, for a known binary vector $\mathbf{x}$ (i.e., a known $\mathcal{S}$)  that satisfies~\eqref{eq:con7}, we arrive at the following optimization problem for the D2D data offloading:
{\color{black}
\begin{align}
\hspace{-0.6mm} (\boldsymbol{\mathcal{P}}_D) :~~ %\left(\mathbf{x}^*,\{\boldsymbol{\Phi}^*(t)\}_{t=0}^{T}\right) =
& \hspace{-5mm} \underset{\{\boldsymbol{\Phi}(t)\}_{t=1}^{T}}{\min} 
\hspace{-1mm} \frac{1}{T} \sum_{t=1}^T \bigg( \alpha \widetilde{F}(t) 
\hspace{-0.6mm} + \hspace{-0.6mm} \beta \hspace{-0.6mm} \sum_{i \in \mathcal{S}} {E}_{i}^{P}(t) \hspace{-0.6mm} + \hspace{-0.6mm} \gamma \hspace{-1mm} \sum_{k \in \mathcal{N}} {E}_{k}^{Tx}(t) \bigg)
\nonumber \\
%\frac{1}{T+1}\sum_{t=0}^{T} \underbrace{\left(\frac{D_{\mathcal{N}}(t)-D_{\mathcal{S}}(t)}{D_{\mathcal{N}}(t)}\right)  \overline{\nabla F(t)}}_{(a)} + \frac{1}{\vert\mathcal{S}\vert} \sum_{i \in \mathcal{S}}\underbrace{ \frac{\gamma}{\sqrt{D_i(t)}}}_{(b)}, \label{eq:obj2}\hspace{-2mm} \\
&  \textrm{s.t.}~ \eqref{eq:con1}-\eqref{eq:con6}, \eqref{eq:con8}-\eqref{eq:con16}.\nonumber
\end{align} }
Since the number of datapoints at the unsampled devices is fixed for all time, $D_{\mathcal{N}}(t)$ can be expressed as $D_{\mathcal{S}}(t) + D_{\hat{\mathcal{S}}}$, where $D_{\hat{\mathcal{S}}} = \sum_{i \in \hat{\mathcal{S}}} D_i$ is a constant. %because the data at the unsampled nodes do not receive offloaded data and are therefore constants. 
%S4 - Training procedure explicitly solving P_D for each X A and S (need to redefine notes in Sec4, do tonight)
Consequently, both the coefficient of $\overline{\nabla F(t)}$ in term $(a)$ and the entirety of term $(b)$ in \eqref{eq:obj_temp} are decreasing functions of the quantity of data $D_i(t)$ at sampled devices $i \in \mathcal{S}$. 
Furthermore, given $\overline{\nabla F(t)}$, both terms $(a)$ and $(b)$ in~\eqref{eq:obj_temp} are convex, via~\eqref{eq:con1} and~\eqref{eq:con2}, with respect to the offloading variables in Problem $\boldsymbol{\mathcal{P}}_D$. %the coefficient in term $(a)$ is convex with respect to total sampled data $D_{\mathcal{S}}(t)$ for all time. As a result, 
The only remaining challenge is then to obtain $\overline{\nabla F(t)}$, which we consider next. %which we address through an iterative approach delineated below. 

\textbf{Sequential gradient approximation:}
Obtaining $\overline{\nabla F(t)}$ requires the knowledge of real-time gradients, $\nabla F(\mathbf{w}_{\mathcal{S}}(t)|\mathcal{D}_i(t))$, $\forall i \in \mathcal{S}$, which are unknown a priori. Furthermore, the gradients of the devices are only observed at the global aggregation time instances $t = k\tau$. Motivated by this, we approximate $\overline{\nabla{{F}(t)}}$ for $t \in \{k\tau+1,\cdots,(k+1)\tau\}$, $ k \in \{1,\cdots,K\}$, using the gradients observed at the most recent global aggregation, i.e., $\nabla F(\mathbf{w}_{\mathcal{S}}(k\tau)|\mathcal{D}_i(k\tau))$, $i\in \mathcal{S}$ on which we perform a sequence of corrective approximations.
%To solve the equivalent optimization problem in \eqref{eq:obj_temp}, we need the gradients $\nabla F(\mathbf{w}_{\mathcal{S}}(t)|\mathcal{D}_i(t))$ ~$\forall i \in \mathcal{S}$. However, in practice, device-to-server and server-to-device communications, which transmit offloading variables $\boldsymbol{\Phi}(t)$ and device gradients $\nabla F(\mathbf{w}_{\mathcal{S}}(t)|\mathcal{D}_i(t))$, $\forall i \in \mathcal{S}$, occur only when $t = k\tau$ where $k \in \{1,...,K \}$. Therefore, the exact gradients are not available at the server, $\forall t \in \{(k-1)\tau,...,k\tau-1\}$ where $k \in \{1,...,K \}$. In the following, we propose an approach to approximate the value of $\overline{\nabla F(t)}$ during the local updates.
% To approximate $\overline{\nabla F(t)}$ during the local updates, we use the information available at the server obtained at the most recent global aggregation, i.e. $\nabla F(\floor{t/\tau})$, 
Specifically, since the average loss $F$ is convex, $\overline{\nabla F(t)}$ is expected to decrease over time. We assume that this decrease occurs linearly and approximate the real-time gradient using the previously observed gradient at the server as $\overline{\nabla F(t)} \approx \overline{\nabla F(k\tau)} / {\alpha^{t-k\tau}_{k+1}}$, $t\in \{k\tau+1,\cdots,(k+1)\tau\}$, $\forall k \in \{1,\cdots,K\}$, where the scaling factor $\alpha_{k+1} \hspace{-.5mm}> \hspace{-.5mm} 1$ is re-adjusted after every global aggregation $k$. Through the re-adjustment procedure, the server receives the gradients and computes the scaling factor for the each aggregation period as $\alpha_{k+1} = \sqrt[\tau]{\overline{\nabla F((k-1)\tau)}/ \overline{\nabla F(k\tau)}}$. %For the neural networks and dataset cpnsidered in this paper, $\alpha_k$ is consistently observed to be in interval $[1.05,1.2]$, $\forall k$.

Given the aforementioned characteristics of terms $(a)$ and $(b)$ in~\eqref{eq:obj_temp}, our proposed iterative approximation of $\overline {\nabla F(t)}$, and the fact that the constraints of $\boldsymbol{\mathcal{P}}_D$ are all affine at each time instance, we can solve this problem as a sequence of convex optimization problems over time. For this, we employ the CVXPY convex optimization software \cite{diamond2016cvxpy}. %\cite{diamond2016cvxpy,agrawal2018rewriting}.

%For each local iteration during the $k$-th global aggregation, i.e., $t \in \{(k-1)\tau + 1, \cdots, k\tau -1 \}$, we compute an approximated global gradient $\overline{\nabla F(t)} $ using the following relationship $\overline{\nabla F(t)} = \nabla F((k-1)\tau) / {\alpha_k}^{t-(k-1)\tau}$. 

%Additionally, as the future costs are not known at the server, there is additional reason for the server to iteratively solve \eqref{eq:obj_temp} at each time instance in order to minimize the instantaneous loss given the current costs and edge weights.

%0.47 originally

\begin{figure*} %\question{need to make this A overline. - 0.48 originally} 
    \centering
    \includegraphics[width=0.87\textwidth]{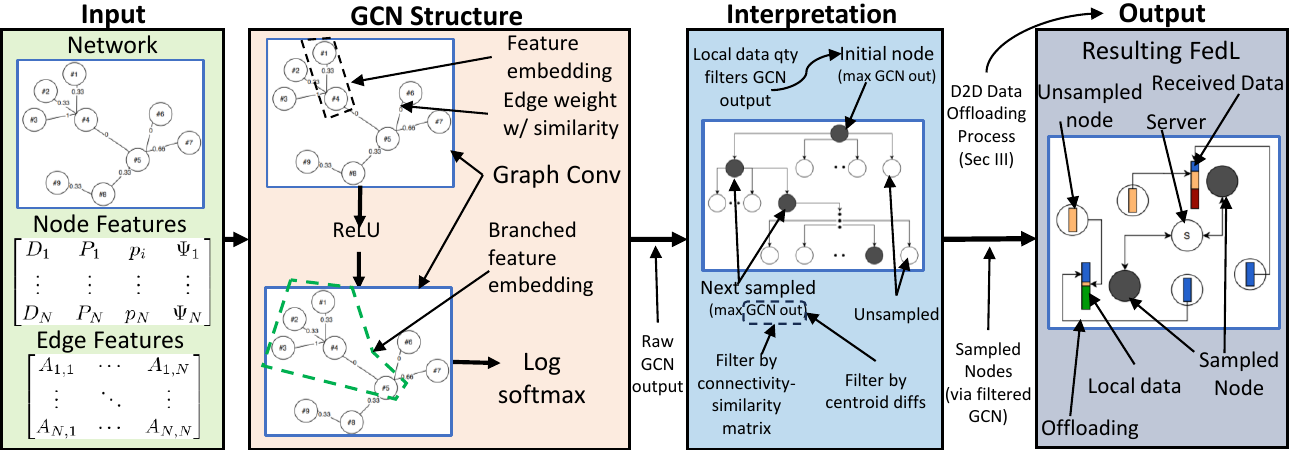}
    \caption{{\color{black}Architecture of our GCN-branch sampling algorithm.
    Given a target network, GCN-branch extracts node and edge features, passing them through two GCN layers, each of which convolves features in local neighborhoods. 
    This process returns raw output probabilities that we filter, based on data quantity, connectivity-similarity matrix, and centroid differences, to obtain a sampled set of nodes in the interpretation stage. 
    With the resulting sampled set, we perform the D2D data offloading process from Sec.~\ref{s:p1}, yielding an output ML model training process for FedL.}}
    \label{fig:gcn_architecture}
    \vspace{-2mm}
\end{figure*}

\section{Smart Device Sampling with D2D Offloading} \label{s:p2}
\noindent %In the following, we first explain the details of the sampling problem in Sec.~\ref{ss:sp2_intro}. We describe the structure of GCNs in Sec.~\ref{ss:gcnmodel}. We explain the details of our sampling technique, which is based on a novel post-processing algorithm applied on GCN outputs, and, then develop a holistic and computationally-efficient framework to solve the complex problem of device sampling with embedded D2D data-offloading over networks of arbitrary size in Sec.~\ref{ss:fullimplement}. We provide a summary of the benefits of our method in Sec.~\ref{ss:gcn_summary}.
We now turn to the sampling decisions in problem $\boldsymbol{\mathcal{P}}$, which must be coupled with the offloading solution to $\boldsymbol{\mathcal{P}}_{D}$. After explaining the rationale for our GCN-based approach (Sec.~\ref{ss:sp2_intro}), we will detail our training procedure encoding the network characteristics (Sec.~\ref{ss:gcnmodel}). Finally, we will develop an iterative procedure for selecting the sample set (Sec.~\ref{ss:fullimplement}).

\subsection{Rationale and Overview of GCN Sampling Approach} \label{ss:sp2_intro}
Sampling the optimal subset of nodes from a resource-constrained network to maximize a utility function (in our case, minimizing the ML loss) has some similarity to  0-1 knapsack problem~\cite{martello2000new}. %, which is a special case in the general family of bin packing problems \cite{man1996approximation}. 
In this combinatorial optimization problem, a set of weights and values for $n$ items are given, where each item can be either added or left out to maximize the value of the items within the knapsack subject to a weight capacity. Analogously, our sampling problem aims to maximize FedL accuracy 
% (by minimizing the difference  $L(w_s(t)|\mathcal{D}_{\mathcal{N}}) - L(w^*|\mathcal{D}_{\mathcal{N}})$)
while adhering to a sampling budget $S = \sum_{i \in \mathcal{N}}x_i$.
% wherein each node $i$ can be sampled or remain unsampled, i.e., $x_i \in \{0,1\}$. 
Strategies for the knapsack problem become unsuitable here because the value that each device provides to FedL is difficult to quantify: it depends on the ML loss function, the gradient descent procedure, and the D2D relationships from Sec.~\ref{s:p1}.
%What makes the knapsack problem unsuitable here is that the value each device provides to FedL is difficult to quantify: it depends on the ML loss function, the gradient descent procedure, and the D2D offloading relationships from Sec.~\ref{s:p1}. %This renders existing techniques for the knapsack problem and its variants (e.g. \cite{frieze1984approximation}) unsuitable for solving $\boldsymbol{\mathcal{P}}$.
%Diverging from knapsack problems, the corresponding value for each node in a FedL setting is unclear and intractable, as one would need to factor in complicated behavior of neural networks, stochastic gradient descent behavior, and D2D relationships such as data similarities and data offloading among devices. This renders the techniques used in literature for the knapsack problem and its variants (e.g. \cite{frieze1984approximation}, \cite{posner1978collapsing}) irrelevant, making our problem very non-trivial and complex. 

To address these complexities, we propose a (separate) ML technique to model the relationship between network characteristics, the sampling set, and the resulting FedL model quality. Specifically, we develop a sampling technique based on active filtering of a Graph Convolutional Network (GCN)~\cite{kipf2016semi}. In a GCN, the learning procedure consists of sequentially feeding an input (the network graph) through a series of graph convolution~\cite{wu2020comprehensive} layers, which generalize the traditional convolution operation into non-Euclidean spaces and thus captures connections among nodes in different neighborhoods.

Our methodology is depicted in Fig.~\ref{fig:gcn_architecture}. GCNs excel at graph-based classification tasks, as they learn over the intrinsic graph structure. 
However, GCNs by themselves have performance issues when there are multiple good candidates for the classification problem~\cite{li2018combinatorial}. 
This holds for our large-scale network scenario, as many high performing sets of sampling candidates can be expected. The data offloading scheme adds another important dimension: a sampled node $i$ may perform poorly when considered in isolation, but it may have high processing capacities $P_i(t)$ and be connected to unsampled nodes $j \in \hat{\mathcal{S}}$ with large quantities of local data and high transfer limits $\Psi_{j,i}(t)$. We address these issues by (i) incorporating the solution from Sec.~\ref{s:p1} into the GCN training procedure, and (ii) proposing \textit{sampling GCN-branch}, a network-based post-processing technique that maps the GCN output to a sampling set by considering the underlying connectivity-similarity matrix.

\begin{figure}[t] 
%\vspace{-2mm}
\centering
\includegraphics[width=0.9\columnwidth]{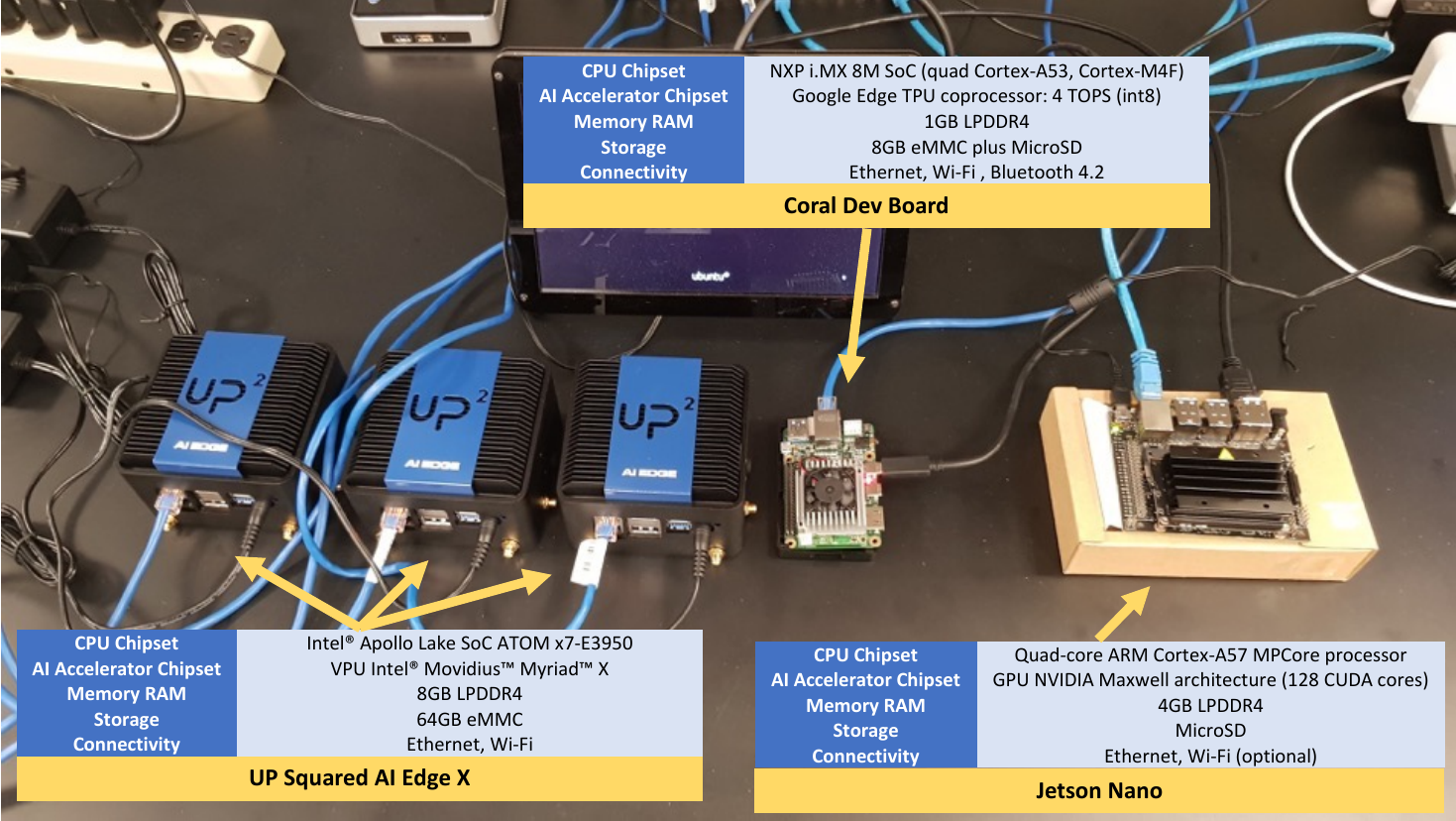}
\caption{IoT testbed used to generate device and link characteristics.}
\vspace{-2mm}
\label{fig:testbed}
\end{figure}

\subsection{GCN Architecture and Training Procedure} \label{ss:gcnmodel}
% Replace $X$ with $\pi$, need to emphasize that each row corresponds to a device
% Need to specify- starting with H^0 = X, we output the graph convolutional layer
% Call them GCN parameter weights for the l-th layer 
% Change W to Q

% \ali{Chris' writing guide:
% SecB - GCN architecture, and how we train the GCN ....
% SecC - Testing, full procedure; with the trained value, do our full netwroks of interest into GCN, and gives us these probabilitie. then we choose from this set
% }
We consider a GCN function $H(\boldsymbol{\pi},\tilde{\mathbf{A}})$ with two inputs: (i)~$\boldsymbol{\pi} \in \mathbb{R}^{N \times U}$, a matrix of $U$ node features, and (ii) $\mathbf{\tilde{A}} \in \mathbb{R}^{N \times N}$, the augmented connectivity-similarity matrix. The feature vector for each node $i$ is defined as  $\boldsymbol{\pi}_i \triangleq [D_i(0), P_i(0), p_i(0), \theta_i(0)]$, forming the rows of $\boldsymbol{\pi}$, and the augmented connectivity-similarity matrix is defined as $\tilde{\mathbf{A}}\triangleq \boldsymbol{{\Lambda}(0)}+\mathbf{I}_{N}$, where $\mathbf{I}_{N}$ denotes the identity matrix~\cite{wu2020comprehensive}. %\cite{zhou2018graph}. 
$H$ consists of two graph convolutional layers separated by a rectified linear unit (ReLU) activation layer~\cite{kipf2016semi}, as depicted in Fig.~\ref{fig:gcn_architecture}. The outputs of each layer are defined as:\vspace{-1mm}
\begin{equation} \label{eq:gcn_layer_params}
    \mathbf{H}^{(l)} \triangleq \sigma\left( \tilde{\mathbf{D}}^{-\frac{1}{2}} \tilde{\mathbf{A}} \tilde{\mathbf{D}}^{-\frac{1}{2}} \mathbf{H}^{(l-1)} \mathbf{Q}^{(l)}\right),~l \in \{1,2\},
    \vspace{-1mm}
\end{equation}
where $\tilde{\mathbf{D}}$ is the degree matrix of $\tilde{\mathbf{A}}$, $\mathbf{Q}^{(l)}$ denotes the trainable weights for the $l$-th layer, and $\sigma$ represents ReLU activation. Note that $\mathbf{H}^{(0)} = \boldsymbol{\pi}$, $\mathbf{Q}^{(1)} \in \mathbb{R}^{U \times O}$, and $\mathbf{Q}^{(2)}\in \mathbb{R}^{O \times 1}$, where $O$ is the dimension of the second layer. Finally, log-softmax activation is applied to $\mathbf{H}^{(2)} \in \mathbb{R}^{N}$ to convert the results into a vector of probabilities, i.e., $\boldsymbol{\Gamma} \in [0,1]^{N}$, representing the likelihood of each node belonging to the sampled set.

\textbf{GCN training procedure:} To train the GCN weights, we generate a set of sample network and node data realizations $e = 1,\cdots,E$ with the properties from Sec.~\ref{sss:devices}\&\ref{sss:graph}. For each realization, we calculate the matrices $\boldsymbol{\pi}_e$ and $\mathbf{\tilde{A}}_e$ corresponding to the inputs of the GCN. Then, for each candidate sampling allocation $\mathbf{x}^s_e=[({x}^s_e)_i]_{1\leq i\leq N}$ (with $\sum_i ({x}^s_e)_i = S$), we solve $\boldsymbol{\mathcal{P}}_{D}$ from Sec.~\ref{s:p1} to obtain the offloading scheme, and then determine the loss of FedL resulting from model training and D2D offloading. Among these, we choose the $\mathbf{x}^{\star}_e$ that yields the smallest objective to be the target GCN output. The collection of $[(\boldsymbol{\pi}_e, \boldsymbol{\tilde{A}}_e, \mathbf{x}^{\star}_e)]_{e=1}^{E}$ form the training samples for the GCN.

%To train the appropriate graph convolutional weights $\mathbf{Q}^{(l)}$, $l \in \{1,2\}$, we generate a set of network realizations. 
% For each network realization, we decide the size for the set of sampled devices, $\mathcal{S}$. 
% by performing a brute force search over all possible sets of sampled devices, followed by D2D offloading and subsequent testing using the sampled devices combined with D2D offloading to train a CNN for classification. 
%Subsequently, GCN is trained using a sequence of training samples, where each sample can be represented as $g\triangleq [[a_g,b_g,c_g],d_g]$, where $a_g$ is the input network, $b_g$ is the device features, $c_g$ is the edge features, and $d_g$ is a single feasible set of sampled nodes determined through brute force search on the network that resulted in the maximum FedL accuracy given the network topology and the device characteristics. To obtain $d_g$ for each unique network, we need to perform $\binom{N}{S}$ instances of FedL with sampling and the corresponding sampled embedded offloading described by Problem $P_D$.% and measure the final model accuracy, $e_g$.
%During this process, the GCN learns the trade-offs among the device features and the overall network architecture to make inferences on selection probabilities.

As the number of devices $N$ increases, the number of choices that will be considered for the sampled set increases combinatorially as $\binom{N}{S}$. An advantage of this GCN procedure is that it can be adapted to networks of varying size: once trained on a set of realizations for tolerable-sized values of $N$, the graph convolutional layer weights $\mathbf{Q}^{(l)}$, $l \in \{1,2\}$, can be applied to the desired network of varying size by shifting through the graph. Our obtained performance results in Sec.~\ref{s:numRes} verify this experimentally.

\subsection{Offloading-Aware Smart Device Sampling} \label{ss:fullimplement} 

\begin{algorithm}[t!] 
   \caption{{\color{black}Summary of Key Steps in Overall Methodology}}
   \label{alg:ovr}
   \begin{algorithmic}[1] 
        {\color{black}
        \STATE // GCN model construction - \textbf{Sec.~\ref{ss:gcnmodel}} 
        \STATE \textit{Step 1:} Generate randomized network realizations, each having graph $G_z$ with $z$ corresponding to a unique realization, for FedL. 
        \STATE \textit{Step 2:} For each realization, find the best sets of sampled devices, e.g., via brute force search, while factoring in optimized data offloading via $(\boldsymbol{\mathcal{P}}_D)$ in \textbf{Sec.~\ref{ss:p1c}}.  
        \STATE \textit{Step 3:} Extract node features (i.e., $P_i, p_i, D_i, \Psi_i$ $\forall i \in \mathcal{N}$) and edge features (i.e., $\psi_{i,j}, A_{i,j}, \lambda_{i,j}$ $\forall i,j \in \mathcal{N}, i \neq j$) for the network realizations. 
        \STATE \textit{Step 4:} Use the results of Step 2 and 3 as the ground truth to train the GCN's parameters, $\boldsymbol{Q}^{(l)}$,~$l \in \{1,2\}$ per~\eqref{eq:gcn_layer_params} in \textbf{Sec.~\ref{ss:gcnmodel}}.
        
        \STATE 
        
        \STATE // Implementation and operation for a given network - follows \textbf{Sec.~\ref{ss:fullimplement}}
        \STATE \textit{Step 1:} Extract node features (i.e., $P_i, p_i, D_i, \Psi_i$ $\forall i \in \mathcal{N}$) and edge features (i.e., $\psi_{i,j}, A_{i,j}, \lambda_{i,j}$ $\forall i,j \in \mathcal{N}, i \neq j$) for a target network, $G = \{\mathcal{N}, \mathcal{E}(t)\}$.
        \STATE \textit{Step 2:} Pass the target network $G = \{ \mathcal{N},\mathcal{E}(t) \}$ and its features (i.e., $P_i, p_i, D_i, \Psi_i, \psi_{i,j}, A_{i,j}, \lambda_{i,j}$ $\forall i,j \in \mathcal{N}, i \neq j$) through the trained GCN, which has trained weights, $\boldsymbol{Q}^{(l)}$,~$l \in \{1,2\}$, in the form of~\eqref{eq:gcn_layer_params} as a result of the procedure in \textbf{Sec.~\ref{ss:gcnmodel}}. 
        \STATE \textit{Step 3:} Perform GCN-branch to obtain the set of sampled devices $\mathcal{S}$, described in \textbf{Sec.~\ref{ss:fullimplement}}. 
        \STATE \textit{Step 4:} Perform the D2D data offloading and resource optimization process of $(\boldsymbol{\mathcal{P}}_D)$ from \textbf{Sec.~\ref{ss:p1c}} to determine the offloading ratios  $\boldsymbol{\Phi}(t)$ for the sampled set $\mathcal{S}$, given some optimization regime $\alpha$, $\beta$, and $\gamma$. 

        \STATE

        \STATE // Operation with device sampling and D2D data offloading - \textbf{Sec.~\ref{s:numRes}}
        \STATE \textit{Step 1:} Network uses GCN-branch to select devices $i \in \mathcal{S}$ as the sampled nodes and $(\boldsymbol{\mathcal{P}}_D)$ to obtain the offloading ratios $\boldsymbol{\Phi}(t)$. 
        \STATE \textit{Step 2:} Sampled devices $i \in \mathcal{S}$ are initialized with ML model parameters $\textbf{w}_i^0$, and locally update their ML model parameters, $\mathbf{w}_i(t)$, via gradient descent in~\eqref{eq:gd}. 
        \STATE \textit{Step 3:} At each training iteration $t$, sampled devices $i \in \mathcal{S}$ receive data $\boldsymbol{\Phi}_{j,i}(t) D_j(t)$ from nearby devices $j \in \mathcal{N}, j \notin \mathcal{S}$. 
        \STATE \textit{Step 4:} At each $k$-th aggregation, the server aggregates the ML model parameters 
        $\mathbf{w}_{\mathcal{S}}(k\tau) = \frac{\sum_{i \in \mathcal{S}}  \Delta_i(k\tau) \mathbf{w}_i(k\tau)}{\sum_{i \in \mathcal{S}} \Delta_i(k\tau)}$,
        and then synchronizes the sampled devices' ML model parameters to $\mathbf{w}_{\mathcal{S}}(k\tau)$. 
        \STATE \textit{Step 5:} Training process continues until $t = T$.
        }
  \end{algorithmic}
\end{algorithm}

Given any network graph, our procedure must solve the sampling problem at the point of FedL initialization, i.e., $t=0$.
With the trained GCN in hand, we obtain $\boldsymbol{\pi}$ and $\tilde{\mathbf{A}}$ for the target network and calculate $\boldsymbol{\Gamma} = H(\boldsymbol{\pi},\tilde{\mathbf{A}})$, $\boldsymbol{\Gamma}=[\Gamma_i]_{1\leq i\leq N}$. 
Given this output, our sampling GCN-branch algorithm populates the set $\mathcal{S}$ as follows. 
Let $\mathcal{N}_p \subset \mathcal{N}$ be the subset of nodes in the 95th percentile of initial data quantity. 
Starting with $\mathcal{S} = \emptyset$, the first node is added according to $S = S \cup \{s_1\}$, where $s_1 = \argmax_{i \in \mathcal{N}_p} \Gamma_i$. 
{\color{black} In this way, the first node sampled is the device with the highest GCN probability among nodes with large local data generation.} 
%the largest data generation nodes.}
To choose subsequent sampled nodes, the algorithm performs a recursive branch-based search on the initial connectivity-similarity matrix $\boldsymbol{\Lambda}(0)$ for nodes with the highest sampling probabilities and the smallest aggregate data similarity, as measured by maximin centroid difference, to the previously sampled nodes.
{\color{black} In doing so, our algorithm filters the search space for the GCN to yield subsequent nodes with the most contribution to the existing sampled set of nodes.
We develop this procedure for two main reasons. Firstly, the GCN scores devices individually, so two devices with high GCN probabilities may not necessarily yield a cohesive fit. Secondly, existing work~\cite{li2018combinatorial,wu2020comprehensive} has shown that a GCN individually may struggle to return a single
optimal set in scenarios where multiple optimal sets are feasible, which is the case in large-scale networks.}

Formally, we choose the $n$-th node addition as $\mathcal{S} = \mathcal{S} \cup \{s_n\}$, where $s_{n} =\argmax_{i \in \mathcal{R}_{s_{n-1}}} \Gamma_i$ 
{\color{black} with $\mathcal{R}_{s_{n-1}}$ denoting the unsampled nodes above the 95th percentile of link dissimilarity to $s_{n-1}$ (i.e., based on $\boldsymbol{\Lambda}(0)$) and the 80th percentile of maximin centroid difference relative to all the nodes in $\mathcal{S}$ thus far.}
% with $\mathcal{R}_{s_{n-1}}$ denoting the neighbor nodes of $s_{n-1}$ within the 95th percentile of link dissimilarity to $s_{n-1}$ (i.e., based on $\boldsymbol{\Lambda}(0)$) and the 80th percentile of maximin centroid difference relative to all the nodes in $\mathcal{S}$ thus far.
% Formally, the $n$-th node addition is $\mathcal{S} = \mathcal{S} \cup \{s_n\}$, where $s_n = \argmax_{i} \Gamma_i $
% we need a set for the recent sampled node $s_j$: $\mathcal{N}_{s_j}$ denotes the nodes within the 98th percentile of dissimilarity to the sampled node $s_j$, and chose the next sampled node as $s_{j+1} =\argmax_{i \in \mathcal{N}_{s_j}} \Gamma_i$
% for the branch, you look at the 98th percentile of the nodes with the least similarity, and then chose the one with the highest GCN probability.
%\argmin_{j\in \hat{\mathcal{S}}} \sum_{i \in \mathcal{S}}\Lambda_{i,j}(0)$. 
%$\argmax_{i}  \argmin_{j\in \hat{\mathcal{S}}} \sum_{k \in \mathcal{S}}\Lambda_{k,j}(0)$
In this way, our branch algorithm relies on the GCN to decide which branch the sampling scheme will follow given its current sampled nodes (visualization in Fig.~\ref{fig:gcn_architecture}), %(see Fig.~\ref{fig:gcn_architecture} for a visualization)
so that subsequent selections are more likely to contain nodes with (i) different %diversities in underlying
data distributions while (ii) leading to new neighborhoods that can contribute to the current set. Once the sampled set $\mathcal{S}$ is determined, the offloading is scheduled for $t = 0,...,T$ per the solver for $\boldsymbol{\mathcal{P}}_{D}$ from Sec.~\ref{s:p1}.

\textbf{Summary of methodology:} Fig.~\ref{fig:BigPicture} summarizes our methodology developed in Sec.~\ref{s:p1}\&\ref{s:p2} for solving $\boldsymbol{\mathcal{P}}$. 
{\color{black} Algorithm~\ref{alg:ovr} also provides an overview of the key steps in the process, focusing specifically on the steps (i) to build the GCN model, (ii) for implementation on a target network, and (iii) for FedL operation on the target network.}
The sequential convex optimization for offloading (Sec.~\ref{s:p1}) is embedded within the GCN-based sampling procedure (Sec.~\ref{s:p2}). Once the model is trained on sample network realizations, it is applied to the target network to generate the $\mathcal{S}$ and $\boldsymbol{\Phi}(t)$ for FedL.

% \begin{itemize}[leftmargin=4mm]
%     \item Our method can be easily generalized to various definitions of the edge weights. It also enjoys \textit{generalizability}, as a result of which it can be once trained on small network sizes and repeatedly applied on larger network sizes.
%     \item Our method \textit{learns} the intractable trade-off between the number of data points at different nodes and data similarity of the neighbors relative to the accuracy of FedL. It further \textit{encapsulates} the network structure and the offloading scheme to infer the best sampled nodes.
%     \item Our method enjoys \textit{flexibility} as it can be easily extended to other scenarios with different features of interest for the nodes and the edges, e.g., defining edge weights based on social tie as in classic social network literature, and prioritizing some sampled nodes due to their transmission capabilities and channel condition.
%     \item Our method enjoys \textit{robustness} to mislabelled data, where specific nodes are a source of noisy data. In this scenario, the server can adjust the edge weights to reflect D2D characteristics and identify sources of noisy data and subsequently filter them out through sampling. 
% \end{itemize}

\vspace{-1mm}
\section{Experimental Evaluation}\label{s:numRes}
% \noindent We now conduct experiments to evaluate our methodology in terms of model accuracy and network resource utilization.
%After presenting our setup in Sec.~\ref{ss:setup}, we discuss our results in Sec.~\ref{ss:sim_baseline}. %explain the experimental setup in Sec.~\ref{ss:setup}, and then present our simulation results in Sec.~\ref{ss:sim_baseline}.

% \begin{figure}[t]
% \centering
% \includegraphics[width=0.45\textwidth]{pix/fmnist.pdf}%[width = 0.45\textwidth,trim = 2cm 2cm 2cm 3cm]{pix/mnist.eps}
% \caption{Testing classification accuracies on F-MNIST for the same setup as in Fig.~\ref{fig:345_mnist}. The results are consistent with the MNIST dataset. The wide margin of improvement obtained by smart sampling with offloading vs. without emphasizes the benefit of considering these two aspects jointly for FedL optimization.}
% \label{fig:345_fmnist}
% %\vspace{-6mm}
% \vspace{-1mm}
% \end{figure}

% testing accuracies of CNNs trained by various sampling schemes, and classic FedL on Fashion-MNIST. Sampling and network sizes are specified above each subplot. Our method always matches or outperforms classic FedL.

\subsection{Setup and Experimental Procedure} \label{ss:setup}
\subsubsection{Network characteristics via wireless testbed}
We employed our IoT testbed in Fig.~\ref{fig:testbed} to obtain device and communication characteristics. It consists of Jetson Nano, Coral Dev Board TPU, and UP Squared AI Edge boards configured to operate in D2D mode.
%, it provides us with heterogeneous communication and computation capabilities in a realistic wireless edge network. 
%We model device characteristics using our IoT testbed depicted in Fig.~\ref{fig:testbed} composed of Jetson Nano, Coral Dev Board TPU, and UP Squared AI Edge boards.
% Each device was designed to fulfill a specific niche in IoT/fog environments, and therefore %Using all three (with varying initial task loads) provides our simulations with genuine measurements of device and communication heterogeneity. Due to its limited resources, we modeled the Coral TPUs as the straggler devices. To determine the effective computing resources and the power consumption, we use Dstat~\cite{dstat}, as the data collection method. We map the computing resources used and the corresponding power consumption at devices to our corresponding costs and capacities (e.g., unit processing costs $\mathbf{p}(t)$) by calculating the (normalized) Gateway Performance Efficiency Factor (GPEF)~\cite{morabito2018legiot}. 
% The benefit of GPEF is that it is a hardware independent metric, which allows us to perform accurate comparisons of the real-world costs of gradient descent or D2D data transfers across heterogeneous networks of devices.
We used Dstat~\cite{dstat} to collect the device resources and power consumption. We map the measured computing resources (in CPU cycles and RAM) and the corresponding power consumption (in mW) at devices to the costs and capacities in our model by calculating the Gateway Performance Efficiency Factor (GPEF)~\cite{morabito2018legiot}. Specifically, to determine the processing costs $p_i(t)$, we measured the GPEF of the devices running gradient iterations on the MNIST dataset~\cite{yann}. For the processing capacities $P_i(t)$, we pushed the devices to 100\% load and measured the GPEF. %(\text{http://yann.lecun.com/exdb/mnist/}), in the range 25\%-75\%,
We initialized the devices at 25\%-75\% loads, and treated the available remaining capacity as the receive buffer parameter $\theta_i(t)$. %the energy consumed and computational resources 

%In practice, we expect other communications between devices to crowd out the available bandwidth. 
For the transmission costs, we measured GPEF spent on D2D offloading over WiFi. Our WiFi links, when only devoted to D2D offloading, consistently saturated at 12 Mbps. To simulate the effect of external tasks, we limit available bandwidth for D2D to 1, 6, and 9 Mbps. We then calculated the transmission resource budget for devices as transfer limits $\Psi_{i}(t)$, and modelled unit transfer costs $\psi_{i,j}(t)$ as normalized D2D latency.

\begin{figure}[!t]
\centering
\includegraphics[width = 0.49\textwidth]{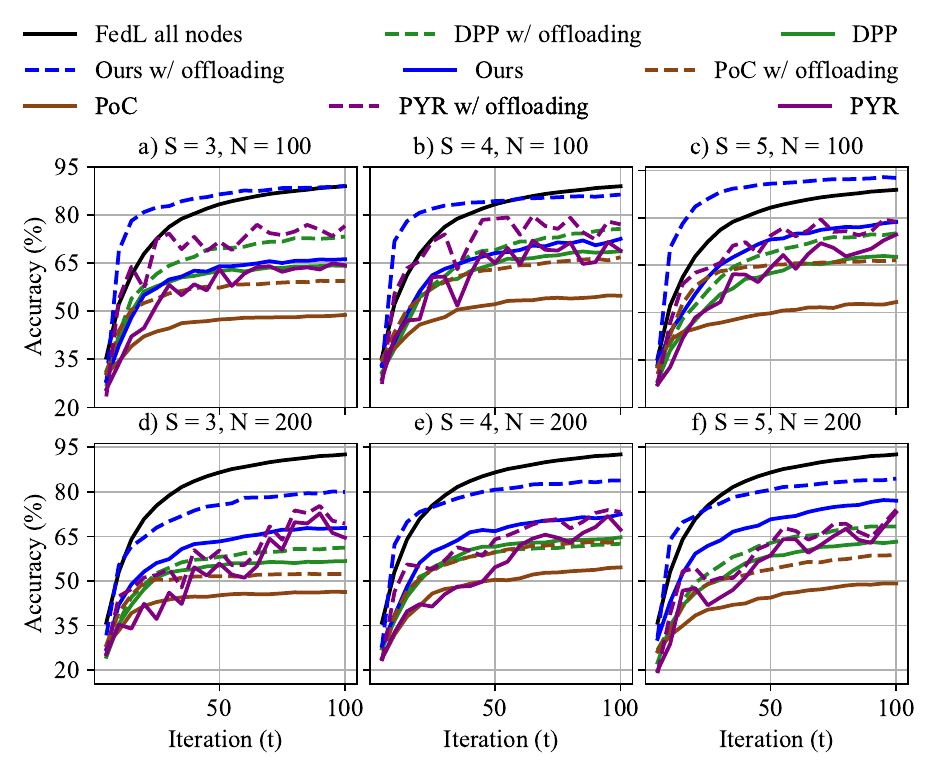}
\caption{{\color{black} Testing classification accuracies over training iterations on MNIST obtained by the sampling schemes with and without offloading, and by FedL using all nodes, for different sampled sizes ($S$) and nodes ($N$). For $S > 3$, our proposed sampling methodology with offloading consistently obtains a wide margin of improvement over all schemes.}}
\label{fig:mnist}
\vspace{-1mm}
\end{figure}

\begin{table}[t]
% \vspace{-4mm}
\caption{{\color{black} Global aggregations required by each scheme to reach a certain percentage of FedL model accuracy on MNIST (69\%) and F-MNIST (41\%) with $S = 6$ and varying $N$. Our sampling methodology with offloading consistently obtains the fastest training time. Dashes indicate that a method was unable to reach the accuracy threshold.}}
%for 6 sampled devices from a network of 600 to reach 85\% and 65\% of FedL accuracy on MNIST and F-MNIST, respectively. Our method consistently attains FedL accuracies with fast convergence.}
\label{tab: mnist}
% {\footnotesize
% \begin{tabularx}{0.48\textwidth}{c *{6}{Y}}
% \toprule[.2em]
% \multirow{2}{*}{\bf{Sampling Method}} & \multicolumn{3}{c}{\bf{Devices (MNIST)}} & \multicolumn{3}{c}{\bf{Devices (F-MNIST)}} \\
% \cmidrule(lr){2-4} \cmidrule{5-7}
% %\cmidule(lr){5-7}
% & \bf{600} & \bf{700} & \bf{800} & \bf{600} & \bf{700} & \bf{800}\\
% \midrule
% Ours w/ offload & 2 & 2 & 5 & 4 & 4 & 3\\
% DPP w/ offload & 13 & 12 & 8 & 15 & \rule{5 mm}{0.2pt} & \rule{5 mm}{0.2pt}\\
% PoC w/ offload & \rule{5 mm}{0.2pt} & 13 & \rule{5 mm}{0.2pt} & 13 & \rule{5 mm}{0.2pt} & 17\\
% \midrule
% Ours & 17 & 6 & 16 & 12 & 6 & 14\\
% DPP & 13 & 17 & 8 & 17 & \rule{5 mm}{0.2pt} & \rule{5 mm}{0.2pt}\\
% PoC & \rule{5 mm}{0.2pt} & \rule{5 mm}{0.2pt} & \rule{5 mm}{0.2pt} & \rule{5 mm}{0.2pt} & \rule{5 mm}{0.2pt} & \rule{5 mm}{0.2pt}\\
% \bottomrule
% \end{tabularx}
% }
{\footnotesize
\begin{tabularx}{0.48\textwidth}{c *{6}{Y}}
\toprule[.2em]
\multirow{2}{*}{\bf{Sampling Method}} & \multicolumn{3}{c}{\bf{Devices (MNIST)}} & \multicolumn{3}{c}{\bf{Devices (F-MNIST)}} \\
\cmidrule(lr){2-4} \cmidrule{5-7}
%\cmidule(lr){5-7}
& \bf{600} & \bf{700} & \bf{800} & \bf{600} & \bf{700} & \bf{800}\\
\midrule 
Ours w/ offload & 4 & 5 & 3 & 3 & 3 & 2 \\
DPP w/ offload & 12 & 12 & 9 & 15 & 15 & 16 \\
PoC w/ offload & \rule{5 mm}{0.2pt} & 14 & \rule{5 mm}{0.2pt} & 13 & 17 & 13 \\
PYR w/ offload & \rule{5 mm}{0.2pt} & 17 & 17 & \rule{5 mm}{0.2pt} & \rule{5 mm}{0.2pt} & \rule{5 mm}{0.2pt}\\
\midrule 
Ours & 10 & 12 & 14 & 10 & 12 & 13\\
DPP & 13 & 16 & 10 & 16 & 18 & 20 \\
PoC & \rule{5 mm}{0.2pt} & \rule{5 mm}{0.2pt} & \rule{5 mm}{0.2pt} & 18 & \rule{5 mm}{0.2pt} & \rule{5 mm}{0.2pt} \\
PYR & \rule{5 mm}{0.2pt} & 19 & 17 & \rule{5 mm}{0.2pt} & \rule{5 mm}{0.2pt} & \rule{5 mm}{0.2pt}\\
\bottomrule
\end{tabularx}
}
\vspace{-1mm}
\end{table}
%\rule{5 mm}{0.2pt}

\begin{figure}[!t]
\centering
\includegraphics[width = 0.49\textwidth]{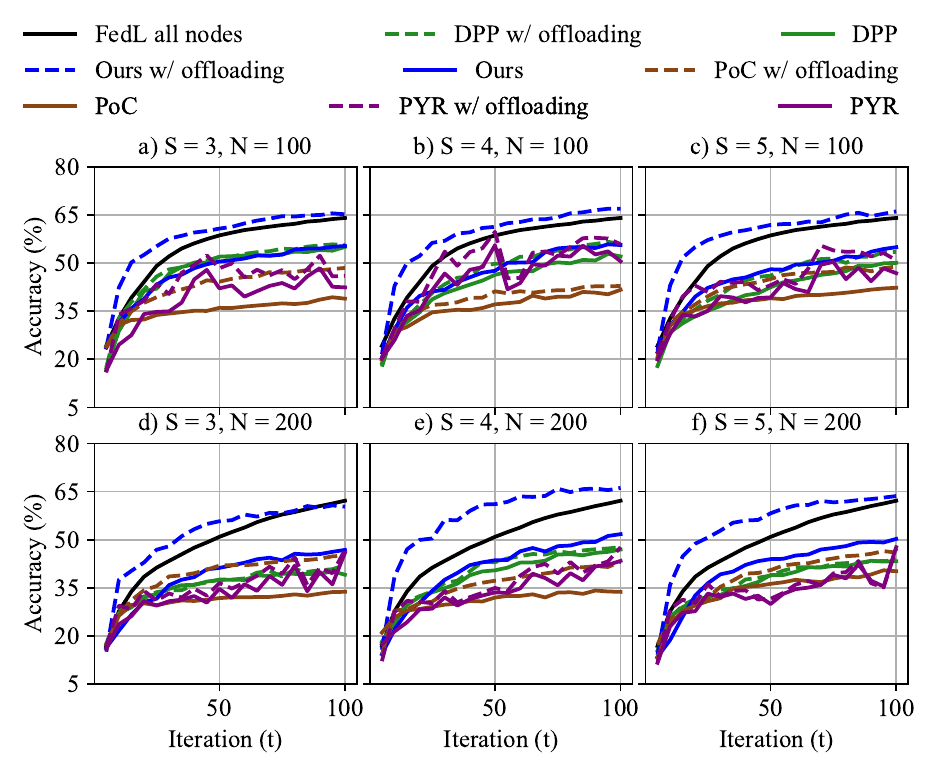}
\caption{{\color{black}Testing classification accuracies on F-MNIST for the same setup as in Fig.~\ref{fig:mnist}. The results are consistent with the MNIST dataset. The wide margin of improvement obtained by our proposed sampling scheme with offloading vs. without emphasizes the benefit of considering these two aspects jointly for FedL optimization.}}
\label{fig:fmnist}
\vspace{-1mm}
\end{figure}

\subsubsection{Datasets and large-scale network generation} \label{sss:datasets}
For FedL training, we use MNIST and Fashion-MNIST (F-MNIST)~\cite{fmnist} image classification datasets. We consider a CNN predictor composed of two convolutional layers with ReLU activation and dropout. The devices perform $\tau = 5$ rounds of gradient descent with a learning rate $\eta = 0.01$. % between global aggregations. %\footnote{https://github.com/zalandoresearch/fashion-mnist}
% As Fashion-MNIST is a drop-in replacement for MNIST, both datasets have $60,000$ training samples and $10,000$ testing samples in the same formats. Therefore, we train CNNs with the same architecture for both datasets. Our CNN itself is composed of two convolutional layers, followed by an activation layer (ReLU) and dropout, and the training procedure to adjust the layer parameters relies on stochastic gradient descent with a constant learning rate $\eta = 0.01$.
Following \cite{wang2021network}, we generate network topologies with $N=100$ to $800$ devices using Erd{\"o}s–Rényi graphs with link formation probability (i.e., $A_{i,j} = 1$) of $0.1$. To produce local datasets across the nodes that are both overlapping and non-i.i.d, the datapoints at each node are chosen uniformly at random with replacement from datapoints among three labels (i.e., image classes). 
Differentiating the labels between devices captures dataset heterogeneity (i.e., from different devices collecting data from different labels). 
The number of initial datapoints $D_i(0)$ at each device follows a normal distribution with mean $\mu=(D_{\mathcal{N}}(0))N^{-1}$ and variance $\sigma^2 = 0.2\mu$. We further estimate the initial similarity weights $\lambda_{i,j}(0)$ based on the procedure discussed in Sec.~\ref{sss:graph}. %In solving the D2D data offloading optimization, we consistently observed $\alpha_k$ to be in the interval $[1.05,1.2]$, $\forall k$.

%the training datasets, and the number of initial datapoints at each device follows a normal distribution with mean $\mu=(D_{\mathcal{N}}(0))^{-1}N$ and variance $0.2\mu$. We further assume that each device is limited to data of three labels and estimate the initial edge weights based on the label differences between devices. In solving the D2D data offloading optimization, we consistently observed $\alpha_k$ to be in the interval $[1.05,1.2]$, $\forall k$.
% Also, to produce data overlaps, we conduct sampling with replacement form the datasets where upon choosing the same datapoint we add a negligible noise to the data-point to incur similarity.
% and data overlap scenarios \ali{this is not good! Why suddenly 800 devices? Why it gives you device heterogeneity?}

For the GCN-based sampling procedure, we train the model on small network realizations of ten devices. We consider sampling budgets of $S = 3$ to $6$, with corresponding training samples $E$ for each case. We save the resulting graph convolutional layer weights $\mathbf{Q}^{(1)}$ and $\mathbf{Q}^{(2)}$ for each choice of $S$ and reapply them on the larger target networks.

% with data relevant to those three labels via sampling with replacement of the relevant training dataset. As a result, devices datasets are populated in a non-i.i.d. manner and have chances of having similar data to other devices.

\subsection{Results and Discussion} \label{ss:sim_baseline}
%\question{Add Chris' analysis: do we discuss in the results why offloading doesn’t really help on top of random in the bottom left (i.e., difference between green curves)?
%interesting that the improvement increases with the size of S, it shows we are taking advantage of the offloading topology between the sampled nodes
%also interesting that the gap for brown decreases as we increase the size of S of course, the key result is that dotted blue is much much better}
% We initially seek to understand the performance of our methodology against those of baselines. 

%\subsubsection{Testing Accuracy:} In Figs.~\ref{fig:345_mnist} and~\ref{fig:345_fmnist}, we perform simulations for six different cases, with the first row sampling three, four, and five devices from a network of 100 devices and the second row repeats the sampling size for a network of 200 devices. Each subplot compares seven different schemes, including FedL with all nodes active.
% , as they train a CNN, and we isolate the performance of pure sampling and that of sampling plus D2D offloading.
In the following experiments, we compare our methodology to several baseline sampling and offloading schemes. 
{\color{black} The four sampling strategies considered are \textit{Ours}, data proportional probabilities (\textit{DPP})~\cite{li2019convergence,karimireddy2020scaffold}, power-of-choice (\textit{PoC})~\cite{pmlr-v151-jee-cho22a}, and PyramidFL (\textit{PYR})~\cite{li2022pyramidfl}.}
%The three sampling schemes are \textit{smart}, \textit{random}, and \textit{heuristic} sampling. 
\textit{DPP} is a random device sampling technique based on the quantity of local data at network devices, with larger datasets incurring a higher probability of sampling for FedL aggregations. On the other hand, \textit{PoC} is a sampling technique based on the instantaneous training losses across network devices, with larger losses equating to a higher probability of being part of the sampled device set.
{\color{black}\textit{PYR} is a time-varying sampling technique based on a combination of exploration and exploitation across network devices, with exploitation referring to the best performing nodes and exploration representing careful selection of other nodes.
For \textit{PYR}, we use an even ratio of $0.5$ for both exploration and exploitation.}
Each of the four sampling schemes is investigated both \textit{with} and \textit{without} data offloading. 
For our sampling methodology, we employ our corresponding data offloading methodology. For \textit{DPP}, \textit{PoC}, and \textit{PYR} sampling methodologies, we perform a greedy offloading that minimizes the communication resource costs while maintaining the same nominal offloaded data quantity as our methodology. 
A baseline of FedL with no sampling (i.e., all nodes active) and no offloading is also included.
% Note that we severely restrict the D2D offloading mechanism by only perform D2D offloading at the first four global aggregations to highlight the effectiveness of our combined methodology.

{\color{black}In the following, we will first investigate classification performances, convergence speeds, and resource utilization of the various sampling methodologies. Then, we will qualitatively examine characteristics of our optimization formulation. Finally, we further examine different optimization regimes (i.e., combinations of $\alpha$, $\beta$, and $\gamma$ on the performance of our methodology), the impact of time-varying links, and integration with FedDrop~\cite{wen2022federated} in the Appendices.}
% We perform further investigations into the various characteristics of our methodology such as 

\begin{figure}[t]
\centering
% \subcaptionbox{MNIST dataset \label{fig:mnist_bar}}[.485\linewidth][c]{%
% \subcaptionbox{FMNIST dataset\label{fig:fmnist_bar}}[.485\linewidth][c]{%
\includegraphics[width=0.96\linewidth]{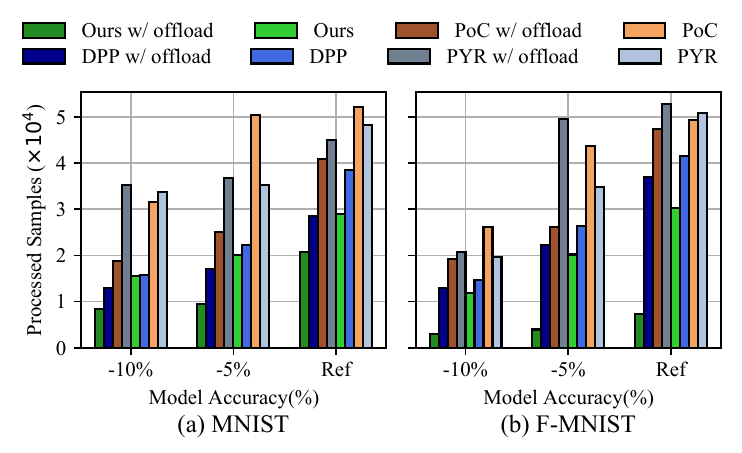}
\caption{{\color{black}Number of samples processed by the schemes with and without offloading to reach within a certain percentage of a reference testing accuracy (Ref, 69\% for MNIST and 41\% for F-MNIST). Our smart sampling with offloading methodology scales the best in terms of data needed for ML model training.}}
\label{fig:bars}
\vspace{-1mm}
\end{figure}

\subsubsection{Model accuracy}
Figs.~\ref{fig:mnist} and~\ref{fig:fmnist} show FedL accuracy for different sampling baselines on two popular ML datasets in six different combinations ($S = 3, 4, 5$ and $N = 100, 200$).
{\color{black} For these experiments, our proposed methodology solves $(\boldsymbol{\mathcal{P}}_D)$ for the data offloading and resource optimization component using $\alpha = 100$, $\beta = 0.001$, and $\gamma = 0.01$ for MNIST or $\gamma = 0.006$ for FMNIST.}\footnote{These values of $\alpha$, $\beta$, and $\gamma$ enable the three objective function terms of ($\boldsymbol{\mathcal{P}}_D$) to all have the same order of magnitude, i.e., $\mathcal{O}(10^{2})$.} 
{\color{black} We further compare the above $\alpha$, $\beta$, and $\gamma$ combination, which we term as the balanced regime, versus the high energy costs regime for $\alpha$, $\beta$, and $\gamma$ in Appendix~\ref{app:high_nrg}, which increases $\beta$ and $\gamma$ by an order of magnitude for both datasets.}

{\color{black}For the experiments involving data offloading, \textit{Ours w/ offloading} consistently demonstrates at least a $5\%$ accuracy improvement over the baseline methodologies for MNIST (Fig.~\ref{fig:mnist}) with a maximum of $13\%$ accuracy improvement in Fig.~\ref{fig:mnist}c).} Similarly on F-MNIST (Fig.~\ref{fig:fmnist}), our methodology is able to continue outperforming the baseline sampling methodologies for various $S$ and $N$ combinations.  
The final accuracy attained by \textit{Ours w/ offloading} is even able to match or outperform \textit{FedL all nodes} in some cases.
Our joint sampling and data offloading methodology is able to achieve better performance due to two main reasons: (i) it minimizes the data skew resulting from unbalanced label frequencies, and (ii) it ensures higher quality of local datasets at sampled nodes, which reduces bias caused by multiple local gradient descents. 
Without offloading, our sampling methodology is also able to either outperform or match the accuracies obtained by the \textit{DPP}, \textit{PoC}, and \textit{PYR} sampling strategies.
% consistently exceeds 20\% for MNIST and 10\% for F-MNIST, which shows that sampling optimization still leads to considerable improvements when D2D is disabled.
Additionally, our method with offloading is able to offer consistent improvements over sampling without offloading in most cases, whereas the effects of data offloading are smaller or even detrimental for the baseline methodologies. This emphasizes the importance of designing the sampling and offloading schemes for FedL jointly.

%As can be seen from Fig.~\ref{fig:345_mnist}, our smart sampling clearly outperforms all other sampling schemes by a wide-margin for all cases using MNIST dataset. When the quantity of sampled devices increases, smart sampling rivals the performance of FedL with all devices active (in subplots b and e) or even outperforms FedL (in subplot c). With the addition of D2D offloading, our unified methodology always matches or exceeds the performance of FedL, and the gap between the testing accuracy of our combined methodology and those of random or heuristic sampling and D2D offloading consistently exceeds 20\%. Also, as can be seen from Fig.~\ref{fig:345_fmnist} using Fashion-MNIST (a much harder classification task), we obtain similar result patterns, where our margins are 10\%. \ali{Let us do average}

\begin{figure}[t]
\centering
\includegraphics[width=0.96\linewidth]{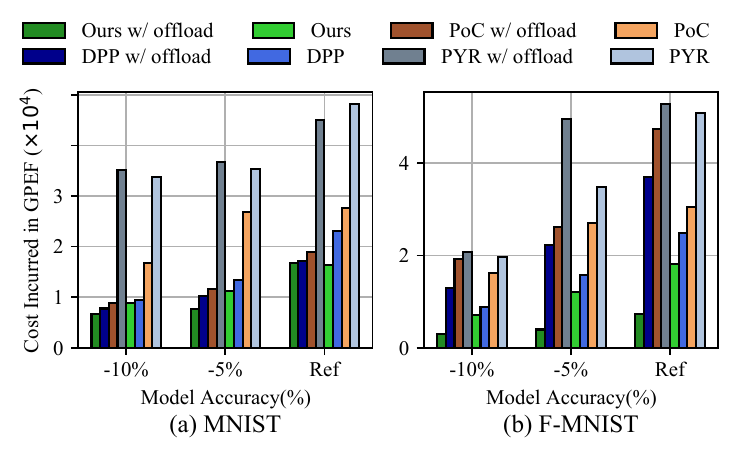}
\caption{{\color{black}The total associated data processing and communication costs measured in GPEF linked to Fig~\ref{fig:bars}. Our proposed sampling methodology with and without offloading obtain some of the lowest costs among all the baseline comparisons regardless of dataset.}}
\label{fig:costs}
\vspace{-1mm}
\end{figure}

% \subsection{Convergence Rate} \label{ss:FedLconv}
\subsubsection{Model convergence speed} We next compare the convergence speeds of our methodology to the other schemes in terms of the number of global aggregations needed to reach specific accuracy thresholds.  %a certain percentage of the final accuracy of FedL with all nodes. 
Table~\ref{tab: mnist} compares the convergence speeds on MNIST (to reach 69\%) and F-MNIST (to reach 41\%), respectively, for $N = 600, 700, 800$ and $S = 6$, and uses a horizontal line to indicate when a methodology is unable to reach the specific accuracy thresholds. We investigate sampling methodologies with and without data offloading. %, using $\alpha=100$, $\beta=0.001$, and $\gamma=0.015$ (MNIST) or $\gamma=0.008$ (F-MNIST) for $(\boldsymbol{\mathcal{P}})$.
% Methods under evaluation are assumed to involve data offloading unless stated otherwise. 

{\color{black} For methodologies involving data offloading, our method obtains the fastest convergences rates. 
Specifically, our method is at least 58\% and 80\% faster, in terms of global aggregations needed to reach the accuracy threshold, than the other baselines on MNIST and F-MNIST respectively.} 
Furthermore, our method is only one to consistently be able to reach the accuracy thresholds. 
For example, while the baseline \textit{PoC w/ offload} is able to reach the threshold for MNIST with $N=700$, it fails to do for $N=600$ and $N=800$ on the same dataset. 

% \begin{figure}[t]
% \centering
% \includegraphics[width = 0.45\textwidth]{pix/tr_main_all.pdf}
% \caption{{\color{black} The average data offloading rates across all D2D links for various network sizes and sampling quantities. The average data offloading rates are naturally lower than in Fig.~\ref{fig:tr_optim_main} as some links may not offload any data.} }
% \label{fig:tr_optim_main_all}
% \vspace{-1mm}
% \end{figure}

Even without data offloading, our sampling methodology is still able to consistently reach the accuracy thresholds for both MNIST and F-MNIST. 
By contrast, both \textit{DPP} and \textit{PoC} baselines are unable to consistently obtain the accuracy thresholds for F-MNIST.
Overall, we see that our joint sampling and offloading methodology obtains faster training speeds than the other methods. 
Enabling offloading is also seen to improve the convergence speeds of each sampling scheme; in fact, without offloading, several baseline cases fail to reach the given percentage of the FedL baseline.

\begin{figure}[t]
\centering
\includegraphics[width = 0.45\textwidth]{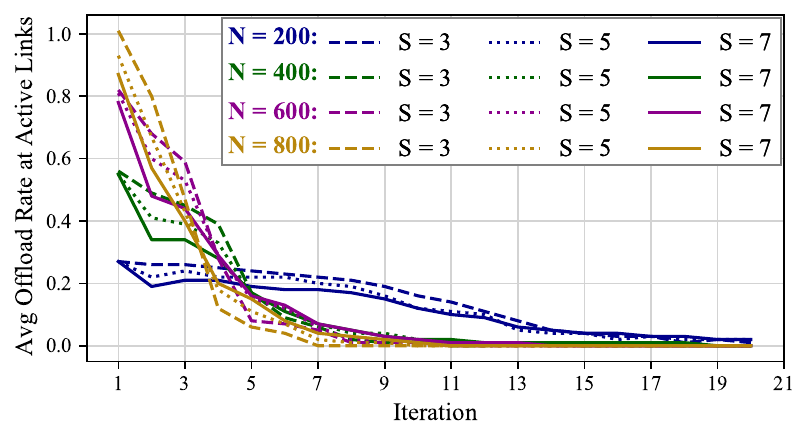}
\caption{The average data offloading rates at active D2D links while varying network sizes and sampling quantities. Larger networks exhibit higher initial average offloading rates and correspondingly have data offloading rates that decay faster in the FedL training process. }
\label{fig:tr_optim_main}
\vspace{-1mm}
\end{figure}

\subsubsection{Resource utilization} Finally, we compare the resource utilization for the different schemes in terms of the total data processed and energy used (measured by GPEF consumption) across a sampled set $S$ of 6 nodes in a network of $N=700$ edge devices. %, using the same $\alpha$, $\beta$, and $\gamma$ for $(\boldsymbol{\mathcal{P}})$ as those in Table.~\ref{tab: mnist}. 
Fig.~\ref{fig:bars} depicts the total data processed while Fig.~\ref{fig:costs} shows the corresponding energy costs in GPEF. 
Both figures compare the resource consumption (data processed or energy used) that sampling methodologies with and without offloading incur in order to reach specific accuracy thresholds on both MNIST and F-MNIST. 
We see that our sampling methodology (both with and without offloading) requires fewer data to be processed and thus results in less energy use relative to the baselines for various accuracy thresholds. % which highlights the computational efficiency obtained by our method. 

\begin{figure*}[t]
\centering
\includegraphics[width = 0.95\textwidth,height=5.5cm]{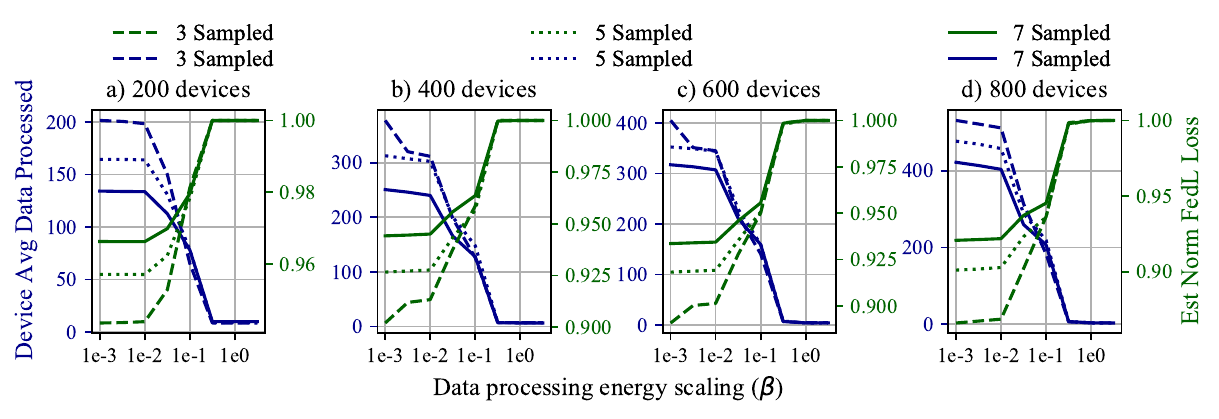} %0.8\textwidth
\caption{Average data processed per sampled device when varying the data processing energy scaling, $\beta$  from~\eqref{eq:obj_new}. 
As $\beta$ increases, networks are disincentivized to perform costly ML training, which leads to lower average data processed per sampled device and correspondingly increases in the FedL loss. 
In addition, as networks grow in size, the average data processed per sampled device increases as more cost-effective D2D links become available.}
\label{fig:p_optim_main}
\vspace{-4mm}
\end{figure*}

\begin{figure*}[t]
\centering
\includegraphics[width = 0.96\textwidth,height=5cm]{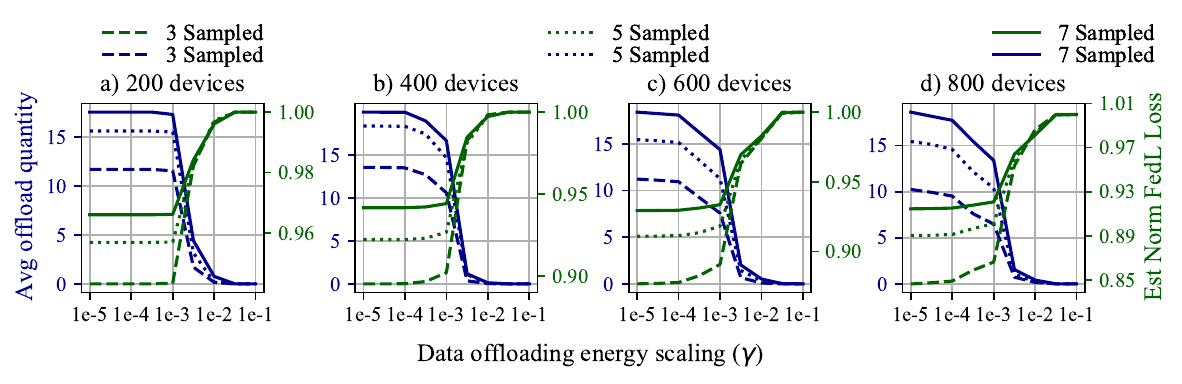} %0.78
\caption{Average data offloaded from unsampled to sampled device when varying the data offloading energy scaling, $\gamma$ from~\eqref{eq:obj_new}. 
As $\gamma$ increases in value, D2D links become more expensive to use, leading to an expected decrease in average offloaded data quantity and a corresponding growth in the FedL loss term. 
The effect of sampling size can also be seen to influence the average data offloading quantity, with smaller sampling quantities yielding lower average data offloading rates as there are physically less links from unsampled to sampled devices.}
\label{fig:tx_optim_main}
\vspace{-4mm}
\end{figure*}

Specifically, on the processed data side (Fig.~\ref{fig:bars}), \textit{Ours w/ offload} processes the least amount of training data across all cases. 
{\color{black}
As the reference accuracy level increases, our methodology consistently requires fewer datapoints compared to the other methods (over 60\% fewer on average), emphasizing its ability to get the most value from each processed datum.}
This emphasizes the benefits of well-designed data offloading in reducing the computational burden across the network as a whole, which is especially valuable in IoT and mobile edge applications. %when computational resources are scarce as 
Furthermore, \textit{Ours} without offloading requires fewer processed data than all other methods under evaluation aside from \textit{Ours w/ offload}, which highlights the resource efficiency of our sampling methodology on its own. 

Our methodologies are able to maintain their advantages on the energy costs side (Fig.~\ref{fig:costs}). Similar to Fig.~\ref{fig:bars}, our methodology with and without data offloading incurs the least amount of energy use as measured in terms of GPEF. 
{\color{black} On average, our method uses less energy than the baseline methods on MNIST and significantly fewer energy than the baseline methods on FMNIST.}
While \textit{Ours w/ offload} processed the least amount of training data in Fig.~\ref{fig:bars}, we see that, for some cases, the energy costs incurred by it are fairly similar to those total costs incurred by \textit{Ours} without data offloading in Fig.~\ref{fig:costs}. 
So even though \textit{Ours w/ offload} may process fewer data and thus have a lower data processing cost, it seems that the data offloading costs incurred over D2D links can, in specific scenarios, balance out the energy savings from fewer processed data for our methodology.

These experiments demonstrate that our joint optimization method exceeds baseline performances in terms of model accuracy, convergence speed, and resource utilization.
Next, we experimentally characterize the behavior of the optimization formulation $(\boldsymbol{\mathcal{P}})$ by investigating its responses over training iterations, and the two energy scaling variables, $\beta$ and $\gamma$.

%We depict the cumulative processed data for all sampling schemes with and without D2D data offloading in Fig.~\ref{fig:bars}. The cumulative processed data allows us to evaluate the efficiency of our method as well as its ability to filter out duplicate and low value data. As the required testing accuracy thresholds increase, the cumulative processed data increase rapidly for all sampling and offloading schemes. It can be seen that our method achieves the lowest rate of increase however, meaning that our method uses fewer computational resources to attain similar accuracies as the other baselines, implying its cost-effectiveness.

%We limit the accuracies on the x-axis for both Figs.~\ref{fig:mnist_bar} and~\ref{fig:fmnist_bar} in order to have an even comparison between all sampling and offloading solutions, as certain sampling and offloading schemes cannot attain higher testing accuracy. Finally, as Fashion-MNIST is a harder classification dataset, all methodologies process more data.

\subsubsection{Optimization characterization}
We characterize the behavior of our optimization formulation $(\boldsymbol{\mathcal{P}})$ in three key aspects: (i) its behavior over time as the FedL process converges, (ii) its sensitivity to the data processing energy use (scaled by $\beta$), and (iii) its responses to the data offloading energy use (scaled by $\gamma$). For all of these experiments, we investigate on networks of $N \in \{200,400,600,800\}$ nodes while sampling $S \in \{3,5,7\}$ devices. In addition, these experiments all use $\alpha=100$, $\beta=0.001$, and $\gamma=0.001$ unless stated otherwise. 

First, we characterize the behavior of $(\boldsymbol{\mathcal{P}})$ with respect to training iterations in Fig.~\ref{fig:tr_optim_main}. % and Fig.~\ref{fig:tr_optim_main_all}.
Here, Fig.~\ref{fig:tr_optim_main} highlights the average offloading rate at active links (i.e., links that have a non-zero data offloading rate). % while Fig.~\ref{fig:tr_optim_main_all} shows the average offloading rate for all available D2D links, including those that have no data offloading. 
One of the major effects shown in Fig.~\ref{fig:tr_optim_main} is device/link saturation, in which D2D links gradually cease to be used for data offloading. The main reason for this is that, over time, sampled devices gather enough data from nearby willing unsampled devices that they no longer have any remaining data processing capabilities (as described in~\eqref{eq:con13}). 
This effect becomes more noticeable as the edge network grows in size, which leads to more D2D connections from unsampled to sampled devices and thus faster device/link saturation. 
For instance in Fig.~\ref{fig:tr_optim_main}, the average data offloading rate at active D2D links rapidly decreases (in less than 10 iterations) from more than $0.8$ to $0$ for $N=800$ regardless of sampled device set size, while for $N=200$, the average data offloading rate requires roughly 20 iterations to approach $0$ for all $S \in \{3,5,7\}$. 
% To summarize, Fig.~\ref{fig:tr_optim_main} shows that, among active links, the average data offloading rate gradually converges to $0$ while Fig.~\ref{fig:tr_optim_main_all} shows that fewer links 
% Both Fig.~\ref{fig:tr_optim_main} and Fig.~\ref{fig:tr_optim_main_all} are consistent in showing that (i) 

More D2D connections among devices also leads to higher initial average data offloading rates among active D2D links in Fig.~\ref{fig:tr_optim_main}. The two reasons for this are: (i) more links leads to more valuable data offloading opportunities and thereby more data offloading, and (ii) simulation setup in Sec.~\ref{sss:datasets} in which the datasets for larger networks are initialized at smaller sizes in order to form a partition of the full ML dataset.

% (i) link saturation (as a result of dataset partition), so for smaller networks, needs more iterations for the data offloading to conclude
% (ii) larger sampled sets have tend larger data offloading rates due to more links from unsampled to sampled devices
% % (iii) larger sampled sets have quicker convergence of offloading rate to 0 because more links means quicker offloading of valuable data to sampled devices
% % \textbf{Behavior of $(\boldsymbol{\mathcal{P}})$ over time:}

Next, we investigate the sensitivity of $(\boldsymbol{\mathcal{P}})$ to $\beta$ in Fig.~\ref{fig:p_optim_main}. 
As we vary $\beta$ from $0.001$ to $10$, the cost of processing data becomes more expensive relative to the FedL loss term in $(\boldsymbol{\mathcal{P}})$ and~\eqref{eq:obj_new}. 
Initially, for small $\beta$ in $N=200$, the cost of processing data is essentially negligible and the average data processed is at a stable point until the $\beta$ scaling reaches a point that is comparable in value to the FedL loss term in~\eqref{eq:obj_new}. 
Further increases in $\beta$ then lead to reductions in the average data processed per device until nearly no data is processed, and, correspondingly, the estimated and normalized FedL loss increases. 

Fig.~\ref{fig:p_optim_main} also highlights some of the effects of large-scale networks and sampled device set size. 
Here, as the networks grow in size from $200$ to $800$ devices, we see that the average data processed per sampled device is increasing, specifically from under $200$ data per device when $N=200$ to over $400$ data per device when $N=800$. This is because larger networks have more links, and therefore more valuable data offloading, which leads to more D2D offloading and more data processed. 
On the other hand, as the sampled device set size increases from $S=3$ to $S=7$, the opposite effect occurs as the average data processed per sampled device decreases. More sampled devices not only leads to more D2D data offloading opportunities, but also enables the unsampled devices to spread out their data to more sampled devices. 
As a result, the burden of data processing can be shared among more sampled devices, thus leading to the decrease in average data processed per sampled device. 

% (i) average data processed increases in larger networks because you have more links, and more valuable data offloading, which leads to more data processed
% (ii) different sampled device set sizes, more sampled leads to lower average data processed
% \textbf{Sensitivity of $(\boldsymbol{\mathcal{P}})$ to $\beta$:}

% \begin{figure*}[t!]
% \centering
% \begin{subfigure}{.49\textwidth}
%   \centering
%   \includegraphics[width=.98\linewidth]{pix/nrg_regime_compmnist.pdf}
%     \vspace{-1mm}
%   \caption{{\color{black} Comparing the high energy costs regime relative to the balanced regime for MNIST. The high energy costs regime uses $\alpha = 100$, $\beta = 0.01$, and $\gamma = 0.1$, while the balanced regime relies on $\alpha = 100$, $\beta = 0.001$, and $\gamma = 0.01$}}
%   \label{fig:high_nrg_mnist_ovr}
% \end{subfigure}
% \begin{subfigure}{.49\textwidth}
%   \centering
%   \includegraphics[width=.98\linewidth]{pix/nrg_regime_compfmnist.pdf}
%     \vspace{-1mm}
%   \caption{{\color{black} Comparing the high energy costs regime relative to the balanced regime for MNIST. The high energy costs regime uses $\alpha = 100$, $\beta = 0.01$, and $\gamma = 0.06$, while the balanced regime relies on $\alpha = 100$, $\beta = 0.001$, and $\gamma = 0.006$}}
%   \label{fig:high_nrg_fmnist_ovr}
% \end{subfigure}
% \vspace{-1.5mm}
% \caption{{\color{black} Comparing the performance of our methodology under the high energy cost regime versus the standard balanced cost regime for MNIST and FMNIST under a variety of network configurations.}}
% \label{fig:ovr_high_nrg_comparison}
% \end{figure*}

Finally, we observe the responsiveness of $(\boldsymbol{\mathcal{P}})$ to changes in the value of D2D offloading costs in Fig.~\ref{fig:tx_optim_main} by measuring the change in average data offload quantity relative to $\gamma$, the scaling of D2D offloading. 
As $\gamma$ increases from $0.00001$ to $0.01$, the cost of D2D data offloading becomes more expensive relative to the data processing and FedL cost terms from~\eqref{eq:obj_new}. 
Small $\gamma$ initially has little impact on average data offloading in small networks of $N=200$ until the $\gamma$ scaling term makes the D2D communication cost comparable to that of the FedL performance term in~\eqref{eq:obj_new}. Then average data offloading quantity rapidly decreases to $0$ and correspondingly the normalized FedL term increases. 

Fig.~\ref{fig:tx_optim_main} provides two additional insights. First, as the network grows in size from $N=200$ to $N=800$, the average data offload quantity per  D2D link remains quite stable, with all Fig.~\ref{fig:tx_optim_main}a)-Fig.~\ref{fig:tx_optim_main}d) demonstrating average data offload rates between $10$ to $20$ for $\gamma=0.00001$. 
That being said, larger networks do demonstrate hightened sensitivity to $\gamma$. For example, for all $S \in \{3,5,7\}$, while the average data offloaded in a network with $N=200$ remains stable for $\gamma \in [0.00001, 0.001]$, the average data offloaded decreases more and more rapidly as the network size grows to $400$, $600$, and then $800$ nodes for the same $\gamma$ range. 
The second is that as the sampled device set size increases from $S=3$ to $S=7$, the average data offloaded per D2D link increases. This is due to the existence of more D2D connections between unsampled and sampled devices, enabling more links to be active. 

% \textbf{Responsiveness of $(\boldsymbol{\mathcal{P}})$ to $\gamma$:}

\section{Conclusion and Future Work}
%We propose the first holistic framework for the device sampling with embedded D2D data offloading for FedL. We revealed the fact that each of the problems are individually non-trivial, and then broken down the problem into two sub-problems. 
%We first propose a theoretical framework to obtain the convergence bound of FedL under a fixed sampling set of device. Then, through a series of approximations on the derived bound, we translate the problem into a series of convex problems. Using the results of D2D offloading, we develop an algorithm, \textit{sampling GCN-branch}, that learns the complicated inter-relationships between network devices when selective D2D data offloading occurs, and solves the problem of device sampling with embedded D2D data offloading for FedL. To evaluate the performance of our methodology, we performed simulations using device and network data from state-of-the-art IoT devices,
%and revealed that our method outperforms classic FedL yet  cumulatively processes significantly less data.
 
%We will look forward to continue to investigate how FedL can be enhanced through embedding network information, such as network stability or link failures, within its solution. 
\noindent In this paper, we developed a novel methodology to solve the joint sampling and D2D offloading optimization problem for FedL. Our theoretical analysis of the offloading subproblem produced new convergence bounds for FedL, and led to a sequential convex optimization solver. We then developed a GCN-based algorithm that determines the sampling strategy by learning the relationships between the network properties, the offloading topology, the sampling set, and the FedL accuracy. Our implementations using popular datasets and real-world IoT measurements from our testbed demonstrated that our methodology obtains significant improvements in terms of datapoints processed, training speed, and resulting model accuracy compared to several other algorithms, including FedL using all devices. 
Future work may consider the integration of further realistic network characteristics such as unlabeled data~\cite{wang2023multi,wagle2022embedding} on FedL, or {\color{black} investigate the integration of realistic wireless conditions, such as fading, shadowing, and interference among others~\cite{amiri2020federated}, for D2D cooperation within large-scale networks for FedL.}

\appendices
\vspace{-1mm}
\section{Proof of Theorem~\ref{thm:error}}\label{app:main1}
\noindent Since $\mathbf{v}_k(t)= \mathbf{v}_k(t-1) - \eta \nabla F(\mathbf{v}_k(t-1)|\mathcal{D}_{\mathcal{N}}(t))$, $\mathbf{w}_{\mathcal{S}}(t)  =  \mathbf{w}_{\mathcal{S}}(t-1) \hspace{-.5mm}- \hspace{-.5mm}\eta (\nabla F(\mathbf{w}_{\mathcal{S}}(t-1)|\mathcal{D}_{\mathcal{N}}(t)) \hspace{-.5mm}+\hspace{-.5mm} \zeta(\mathbf{w}_{\mathcal{S}}(t-1))$, we get:

\vspace{-5mm}
\small
\begin{align} \label{eq:t1_p1}
&\hspace{-2mm}\big\Vert \mathbf{w}_{\mathcal{S}}(t) - \mathbf{v}_{k}(t) \big\Vert \hspace{-1mm} = \hspace{-1mm} \Big\Vert \mathbf{w}_{\mathcal{S}}(t \hspace{-0.5mm}- \hspace{-0.5mm} 1) - \mathbf{v}_k(t\hspace{-0.5mm}-\hspace{-0.5mm}1) - \eta \zeta(\mathbf{w}_{\mathcal{S}}(t\hspace{-0.5mm}-\hspace{-0.5mm}1)) \nonumber\\
&\hspace{-2mm}- \eta\nabla F(\mathbf{w}_{\mathcal{S}}(t-1)|\mathcal{D}_{\mathcal{N}}(t)) +\eta \nabla F(\mathbf{v}_k(t-1) \vert \mathcal{D}_{\mathcal{N}}(t)) \Big\Vert.\hspace{-2mm}
\end{align}
\normalsize

\noindent We simplify \eqref{eq:t1_p1} through the following steps:
\vspace{-5mm}

\small
\begin{align}\label{eq:t1_bigboy}
& \Vert \mathbf{w}_{\mathcal{S}}(t) - \mathbf{v}_{k}(t) \Vert \overset{(a)}{\leq} \Vert \mathbf{w}_{\mathcal{S}}(t-1)-\mathbf{v}_k(t-1) \Vert \nonumber \\
&  + \eta \sum_{i \in \mathcal{N}} \frac{D_i(t-1)}{D_{\mathcal{N}}(t)} \Vert\nabla F(\mathbf{w}_{\mathcal{S}}(t-1)|\mathcal{D}_i(t-1)) \nonumber \\
& - \nabla F(\mathbf{v}_k(t-1)|\mathcal{D}_i(t-1)) \Vert +\eta \Vert \zeta(\mathbf{w}_{\mathcal{S}}(t-1)) \Vert \nonumber \\
&\overset{(b)}{\leq} \Vert \mathbf{w}_{\mathcal{S}}(t-1)-\mathbf{v}_k(t-1) \Vert + \eta \Vert \zeta(\mathbf{w}_{\mathcal{S}}(t-1))\Vert  \nonumber \\ 
&+\eta\beta \sum_{j \in \mathcal{N}}\frac{D_j(t-1)}{D_{\mathcal{N}}(t) D_{\mathcal{S}}(t) } \sum_{i \in \mathcal{S}}D_i(t-1)\Vert \mathbf{w}_{\mathcal{S}}(t-1) - \mathbf{v}_k(t-1)\Vert \nonumber \\
&\overset{(c)}{\leq} \Vert \mathbf{w}_{\mathcal{S}}(t-1) - \mathbf{v}_k(t-1) \Vert + \eta \Vert \zeta(\mathbf{w}_{\mathcal{S}}(t-1))\Vert \nonumber \\ 
&+\sum_{j \in \mathcal{N}}\frac{D_j(t-1)}{D_{\mathcal{N}}(t) D_{\mathcal{S}}(t) } \sum_{i \in \mathcal{S}}D_i(t-1)\frac{\delta_i(t)}{\beta}(2^{t-1-(k-1)\tau}-1)  \nonumber \\
&\overset{(d)}{\leq} \Vert \mathbf{w}_{\mathcal{S}}(t-1)-\mathbf{v}_k(t-1) \Vert + \frac{1}{\beta} \Vert \zeta(\mathbf{w}_{\mathcal{S}}(t-1))\Vert  \nonumber \\ 
&+\sum_{j \in \mathcal{N}}\frac{D_j(t-1)}{D_{\mathcal{N}}(t)} \frac{\delta_{\mathcal{S}}(t)}{\beta}(2^{t-1-(k-1)\tau}-1), 
\end{align}
\normalsize
\vspace{-4mm}

\noindent where $(a)$ results from expanding $\nabla F(\mathbf{v}_k(t-1)|\mathcal{D}_{\mathcal{N}}(t))$ and applying the triangle inequality repeatedly, $(b)$ follows from using the $\beta$-smoothness of the loss function and the triangle inequality, $(c)$ applies Lemma 3 from \cite{wang2019adaptive}, and $(d)$ uses the expanded form of $\delta_{\mathcal{S}}(t)$ in \eqref{th1:3}. We then rearrange~\eqref{eq:t1_bigboy}:
\begin{equation}
\begin{aligned}
&\Vert \mathbf{w}_{\mathcal{S}}(t)-\mathbf{v}_k(t) \Vert - \Vert \mathbf{w}_{\mathcal{S}}(t-1)-\mathbf{v}_k(t-1) \Vert  \\ 
&\leq \frac{\Upsilon(t,k)}{\beta} +\frac{1}{\beta} \Vert \zeta(\mathbf{w}_{\mathcal{S}}(t-1)) \Vert.
\end{aligned}
\end{equation}
%At global aggregations time instances ($t=k\tau$), we synchronize $\mathbf{v}_{k+1}(t)$ to the result of global aggregation $\mathbf{w}_{\mathcal{S}}(k\tau)$. 
\noindent 
Since $\Vert \mathbf{w}_{\mathcal{S}}(t) - \mathbf{v}_{k}(t)\Vert = 0$ when re-synchronization occurs at $t = k\tau$ , ~$\forall k \in \{1,\cdots,K \}$, we express $\Vert \mathbf{w}_{\mathcal{S}}(t)-\mathbf{v}_k(t) \Vert$ as:
\small
\begin{align}
&\Vert \mathbf{w}_{\mathcal{S}}(t)- \mathbf{v}_k(t) \Vert  = \sum_{y = (k-1)\tau+1}^{t}  \Vert \mathbf{w}_{\mathcal{S}}(y) - \mathbf{v}_k(y) \Vert -  \Vert \mathbf{w}_{\mathcal{S}}(y-1)\nonumber \\ &-  \mathbf{v}_k(y-1) \Vert 
 \leq \frac{1}{\beta} \hspace{-2mm} \sum_{y = (k-1)\tau+1}^{t} \left(\Upsilon(y,k) +\Vert \zeta(\mathbf{w}_{\mathcal{S}}(y-1)) \Vert \right). \label{eq:sample}\hspace{-2mm}
\end{align}
\normalsize

%Proof of Corollary~\ref{c1}
\vspace{-1mm}
\section{Proof of Corollary~\ref{c1}}\label{app:c1}
\noindent We first define $\theta_k(t) \triangleq F(\mathbf{v}_k(t))-F(\mathbf{w}^*(t)) \geq \epsilon$.
% \begin{equation} \label{eq:t2_prelim}
% \begin{aligned} 
% &\frac{1}{\theta_{k}(k\tau)} - \frac{1}{\theta_{k}((k-1)\tau)} = \sum_{t=(k-1)\tau}^{k \tau -1} \frac{1}{\theta_{k}(t+1)} - \frac{1}{\theta_k(t)}\\
% & \geq \tau \xi \eta \left(1-\frac{\beta \eta }{2}\right).
% \end{aligned}
% \end{equation}
Since $F$ is $L$-Lipschitz, we apply the result of Theorem 1, and Lemmas 2 and 6 from \cite{wang2019adaptive} to obtain:
\vspace{-1.2mm}
\small
\begin{align}
% &\frac{1}{\theta_{\hat{K}+1}(t)} - \frac{1}{\theta_{1}(0)} = \left(\frac{1}{\theta_{\hat{K}+1}(t)} - \frac{1}{\theta_{k+1}(\hat{K}\tau)}\right) \nonumber \\
& \theta_{\hat{K}+1}(t)^{-1} - \theta_{1}(0)^{-1} = \left(\theta_{\hat{K}+1}(t)^{-1} - \theta_{k+1}(\hat{K}\tau)^{-1} \right) \nonumber \\
& + \left(\theta_{\hat{K}+1}(\hat{K}\tau)^{-1} - \theta_{\hat{K}}(k\tau)^{-1} \right) + \left(\theta_{\hat{K}}(\hat{K}\tau)^{-1} - \theta_{1}(0)^{-1} \right)\nonumber\\
& \geq \left((t-\hat{K}\tau)\xi \eta \left(1-\frac{\beta \eta}{2}\right)\right) + \hat{K}\tau \xi \eta \left(1- \frac{\beta \eta}{2}\right) - \frac{L}{\beta \epsilon^2} \hat{\Upsilon}(\hat{K}) \nonumber\\
& - (\hat{K}-1) \frac{L}{\beta \epsilon^2}\hat{\Upsilon}(\hat{K}) 
= t\xi \eta \left(1-\frac{\beta \eta}{2}\right) - \hat{K} \frac{L}{\beta \epsilon^2} \hat{\Upsilon}(\hat{K}), \label{eq:c1eq1}\hspace{-3mm}
\end{align}
\normalsize
\noindent where $\hat{\Upsilon}(\hat{K}) \triangleq \sum_{y=(\hat{K}-1)\tau+1}^{\hat{K}\tau} \Upsilon(\hat{K},y) + \Vert \zeta(\mathbf{w}_{\mathcal{S}}(y-1))\Vert $. Using the assumptions from Lemma~2 of~\cite{wang2019adaptive} and the fact that $t \geq \hat{K}\tau$, the RHS of~\eqref{eq:c1eq1} is strictly positive, implying that $\hat{\Upsilon} (\hat{K}) \ll 1$ as $t\xi \eta (1-\frac{\beta \eta}{2})$ is very small. Next, defining $\varrho(t) \triangleq F(\mathbf{w}_{\mathcal{S}}(t)|\mathcal{D}_{\mathcal{N}}(t)) - F(\mathbf{w}^*(t)|\mathcal{D}_{\mathcal{N}}(t)) \geq \epsilon$, we get:

\small
\begin{equation} \label{eq:cleq2}
\hspace{-0mm}
\begin{aligned}
&\varrho(t)^{-1} - \frac{1}{\theta_{\hat{K}+1}(t)}  = \frac{F(\mathbf{v}_{\hat{K}}(t)|\mathcal{D}_{\mathcal{N}}(t)) - F(\mathbf{w}_{\mathcal{S}}(t)|\mathcal{D}_{\mathcal{N}}(t))}{\theta_{\hat{K}+1}(t) \varrho(t) } \\
& \geq \frac{-L}{\beta \epsilon^2} \sum_{y=(\hat{K}-1)\tau+1}^{\hat{K}\tau} \Upsilon(\hat{K},y) + \Vert \zeta(\mathbf{w}_{\mathcal{S}}(y-1)) \Vert.
\end{aligned}
\hspace{-6mm}
\end{equation}
\normalsize

\noindent Combining \eqref{eq:c1eq1} and \eqref{eq:cleq2}, and since $\theta_1(0) > 0$, we get $\varrho(t)^{-1} \geq t \xi \eta \left(1 - \frac{\beta \eta}{2}\right) - \frac{(\hat{K}+1)L}{\beta \epsilon^2} \hat{\Upsilon}(\hat{K})$, taking the reciprocal of which leads to \eqref{eq:cl1_result}.
% \small
% \begin{align}
% &\varrho(t)^{-1} \geq t \xi \eta \left(1 - \frac{\beta \eta}{2}\right) - \frac{(\hat{K}+1)L}{\beta \epsilon^2} \hat{\Upsilon}(\hat{K}). \label{eq:c1e3}
% %& \geq \frac{1}{F(\mathbf{w}_{\mathcal{S}}(t)|\mathcal{D}_{\mathcal{N}}(t)) - F(\mathbf{w}^*|\mathcal{D}_{\mathcal{N}}(t))} - \frac{1}{\theta_{1}(0)} \geq t \xi \eta \left(1 - \frac{\beta \eta}{2}\right) \nonumber  \\
% %&t \xi \eta \left(1 - \frac{\beta \eta}{2}\right) - \frac{(\hat{K}+1)L}{\beta \epsilon^2} \hat{\Upsilon}(\hat{K}). \label{eq:c1e3} %\sum_{y=(\hat{K}-1)\tau+1}^{\hat{K}\tau} \Upsilon(\hat{K},y) + \Vert \zeta(\mathbf{w}_{\mathcal{S}}(y-1))\Vert.\label{eq:cleq3}
% \end{align}
% \normalsize

% Taking the reciprocal of \eqref{eq:c1e3} leads to \eqref{eq:cl1_result}.

%Letting $\hat{\Upsilon}(\hat{K}) \triangleq \sum_{y=(\hat{K}-1)\tau+1}^{\hat{K}\tau} \Upsilon(\hat{K},y) + \Vert \zeta(\mathbf{w}_{\mathcal{S}}(y-1))\Vert $ and taking the reciprocal of \eqref{eq:cleq3}, leads to \eqref{eq:cl1_result}.

%\iffalse
%Proof of Proposition~\ref{prop:1}
\vspace{-1mm}
\section{Proof of Proposition~\ref{prop:1}}\label{app:prop1}
% \noindent Applying triangle inequality on $\Vert \nabla F(\mathbf{w}_{\mathcal{S}}(t)|\mathcal{D}_i(t)) - \nabla F(\mathbf{w}_{\mathcal{S}}(t)|\mathcal{D}_{\mathcal{N}}(t)) - \zeta(\mathbf{w}_{\mathcal{S}}(t)) \Vert $, we get:

% \small
% \begin{equation}
%     \hspace{-3mm}
% \begin{aligned}
% &\Vert\nabla F(\mathbf{w}_{\mathcal{S}}(t)|\mathcal{D}_i(t)) - \nabla F(\mathbf{w}_{\mathcal{S}}(t)|\mathcal{D}_{\mathcal{N}}(t)) - \zeta(\mathbf{w}_{\mathcal{S}}(t))\Vert \leq \\
% &  \Vert\nabla F(\mathbf{w}_{\mathcal{S}}(t)|\mathcal{D}_i(t)) - \nabla  F(\mathbf{w}_{\mathcal{S}}(t)|\mathcal{D}_{\mathcal{N}}(t))\Vert +\Vert \zeta(\mathbf{w}_{\mathcal{S}}(t))\Vert.
% \end{aligned}
% \hspace{-3mm}
% \end{equation}
% \normalsize
\noindent Since $\nabla F(\mathbf{w}_{\mathcal{S}}(t)|\mathcal{D}_i(t))$ is the average of $\nabla F(\mathbf{w}_{\mathcal{S}}(t),x_d,y_d)$, $\forall (x_d,y_d) \in \mathcal{D}_i(t)$, we apply the central limit theorem to view $\nabla F(\mathbf{w}_{\mathcal{S}}(t)|\mathcal{D}_i(t))$ as $D_i(t)$ samples of $\nabla F(\mathbf{w}_{\mathcal{S}}(t),x_d,y_d)$ from a distribution with mean $\nabla F(\mathbf{w}_{\mathcal{S}}(t)|\mathcal{D}_{\mathcal{N}}(t))$. Then, \eqref{eq:lemma} can be upper bounded using the definition in~\eqref{def:e}:
\vspace{-3mm}

\small
\begin{align}
&\Vert\nabla F(\mathbf{w}_{\mathcal{S}}(t)|\mathcal{D}_i(t)) - \nabla F(\mathbf{w}_{\mathcal{S}}(t)|\mathcal{D}_{\mathcal{N}}(t)) - \zeta(\mathbf{w}_{\mathcal{S}}(t))\Vert \nonumber \\
&\leq \Vert \zeta(\mathbf{w}_{\mathcal{S}}(t)) \Vert  + \frac{\gamma}{\sqrt{D_i(t)}} \leq \Bigg\Vert \frac{-1}{D_{\mathcal{N}}(t)} \sum_{i \in \hat{\mathcal{S}}} D_i(t) \nabla F(\mathbf{w}_{\mathcal{S}}(t)|\mathcal{D}_i(t))\nonumber  \\
& + \frac{D_{\mathcal{N}}(t)-D_{\mathcal{S}}(t)}{D_{\mathcal{N}}(t) D_{\mathcal{S}}(t)} \sum_{i \in \mathcal{S}} D_i(t) \nabla F(\mathbf{w}_{\mathcal{S}}(t)|\mathcal{D}_i(t))   \Bigg\Vert + \frac{\gamma}{\sqrt{D_i(t)}}.
\hspace{-4mm}
\end{align}
\normalsize
Applying the triangle inequality on the above gives the result. % produces:
% \begin{equation}
% \hspace{-2mm}
% \begin{aligned} \label{eq:l1-result}
% &\Vert\nabla F(\mathbf{w}_{\mathcal{S}}(t)|\mathcal{D}_i(t)) - \nabla F(\mathbf{w}_{\mathcal{S}}(t)|D_{\mathcal{N}}(t)) - \zeta(\mathbf{w}_{\mathcal{S}}(t))\Vert  \\
% & \leq  \left(\frac{D_{\mathcal{N}}(t)-D_{\mathcal{S}}(t)}{D_{\mathcal{N}}(t)}\right)\overline{\nabla F(t)} + \frac{\gamma}{\sqrt{D_i(t)}} + C.
% \end{aligned}
% \hspace{-3mm}
% \end{equation}

% \vspace{-1mm}

% \begin{figure*}[t]
% \centering
% \begin{subfigure}{.49\textwidth}
%     \centering
%     \includegraphics[width=.98\linewidth]{pix/feddrop_comp_mnist.pdf}
%     \vspace{-1mm}
%     \caption{{\color{black}Leveraging FedDrop~\cite{wen2022federated} as a plug-and-play into our proposed device sampling methodology on MNIST for various combinations of network size and sampled set size.}}
%     \label{fig:feddrop_mnist_ovr}
% \end{subfigure}
% \begin{subfigure}{.49\textwidth}
%     \centering
%     \includegraphics[width=.98\linewidth]{pix/feddrop_comp_fmnist.pdf}
%     \vspace{-1mm}
%     \caption{{\color{black}Examining the effectiveness of FedDrop~\cite{wen2022federated} as a plug-and-play with our proposed device sampling methodology on the FMNIST dataset for various settings.}} 
%     \label{fig:feddrop_fmnist_ovr}
% \end{subfigure}
% \vspace{-1.5mm}
% \caption{{\color{black} Our sampling methodology evaluated both with and without the FedDrop~\cite{wen2022federated} framework integrated within.}}
% \label{fig:feddrop_ovr}
% \end{figure*}

%%%% additional qualitative experiments

\section{{\color{black}High Energy Costs Regime on Performance}}
\label{app:high_nrg}

{\color{black}We investigate the behavior of our methodology in a scenario with energy restrictions, and compare it relative to the balanced regime. 
In the balanced regime, the three objective function terms in $(\boldsymbol{\mathcal{P}}_D)$ are expected to have a similar degree of importance, which we can emulate by having similar order of magnitude via careful choice of $\alpha$, $\beta$, and $\gamma$. 
By contrast, if devices were disconnected from stable power, they would have energy limitations, and therefore the energy cost for data processing and transmissions is likely to be more restrictive. 
In response, a network may increase $\beta$ and $\gamma$ by an order of magnitude, thus entering the high energy cost regime. 

\begin{figure}[t]
\centering
\includegraphics[width=.82\linewidth]{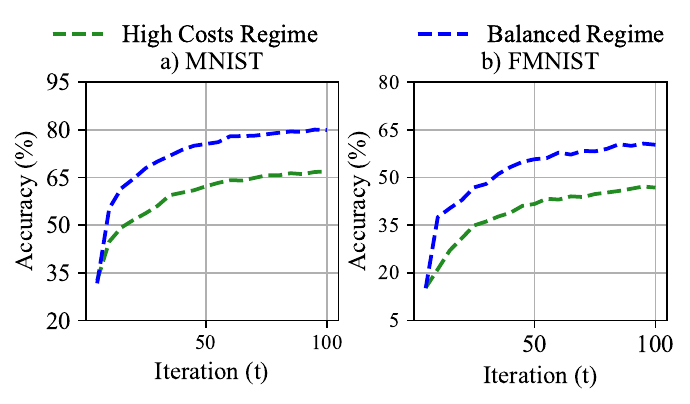}
\caption{{\color{black} Comparing the performance of our methodology under the high energy cost regime versus the standard balanced cost regime for MNIST and FMNIST for a network with $N=200$ and $S=3$.}}
\label{fig:ovr_high_nrg_comparison}
\end{figure}

\begin{table}[t]
\caption{{\color{black} Optimization values of $(\boldsymbol{\mathcal{P}}_D)$ for a network with $N=200$ and $S=3$ on MNIST and FMNIST.
We examine the values of the un-scaled terms of $(\boldsymbol{\mathcal{P}}_D)$, i.e., the averaged estimated FL loss, the data processing energy, and the data transmission energy, for both the balanced regime and the high energy costs regime.}}
{\footnotesize
\begin{tabularx}{0.48\textwidth}
{m{8em} m{3em} m{5em} m{3em} m{5em}}
% {c c c c c}
\toprule[.2em]
& \multicolumn{2}{c}{\textbf{MNIST}} & \multicolumn{2}{c}{\textbf{FMNIST}} \\
\cmidrule(lr){2-3} \cmidrule(lr){4-5} 
& Balanced & High Costs & Balanced  & High Costs \\
\midrule
Estimated Loss & 155.54 & 156.05 & 155.76 & 156.05 \\ 
Data Processing Energy & 120.79 & 99.87 & 148.36 & 142.03 \\
Data Transmission Energy & 188.84 & 0 & 128.47 & 0 \\ 
\bottomrule
\end{tabularx}}
\label{tab:3regimes}
\end{table}

We examine these two regimes in Fig.~\ref{fig:ovr_high_nrg_comparison} and Table~\ref{tab:3regimes}. 
For the balanced regime, we employ $\alpha = 100$, $\beta = 0.001$, and $\gamma = 0.01$ for MNIST or $\gamma = 0.006$ for FMNIST as in Sec.~\ref{ss:sim_baseline}, while, for the high energy costs regime, we will use $\alpha = 100$, $\beta = 0.01$, and $\gamma = 0.1$ for MNIST or $\gamma = 0.06$ for FMNIST.

As the high energy cost regime places greater emphasis on reducing the data processing and D2D data offloading, we see a reduction in the data processing and transmission energies in Table~\ref{tab:3regimes} for the high energy cost regime relative to the balanced regime. 
This reduced data processing and transmission leads to a corresponding increase in the estimated loss in Table~\ref{tab:3regimes}, which is reflected by the decreased classification accuracies shown in Fig.~\ref{fig:ovr_high_nrg_comparison} for the high energy cost regime relative to the balanced regime.}

\section{{\color{black}Impact of Time-Varying Links}}
\label{app:tlinks}

\begin{table}[t]
\caption{{\color{black}Examining the impact of time-varying inactive links (i.e., adjacency matrix changes over time) on the performance of the proposed sampling with D2D data offloading methodology for a network with $N=200$ and $S=3$.}}
{\footnotesize
\begin{tabularx}{0.48\textwidth}
{m{5em} m{4em} m{5em} m{4em} m{5em}}
% {c c c c c}
\toprule[.2em]
& \multicolumn{2}{c}{\textbf{MNIST}} & \multicolumn{2}{c}{\textbf{FMNIST}} \\
\cmidrule(lr){2-3} \cmidrule(lr){4-5}
Link Failure Rate & Final Acc (\%) & Links Available & Final Acc (\%) & Links Available \\
\midrule
0\%  & 79.93 & 7694 & 60.32 & 7699 \\ 
25\%  & 77.96 & 5706 & 59.19 & 5714 \\ 
50\%  & 73.53 & 3788 & 57.97 & 3792 \\ 
75\%  & 73.26 & 1891 & 53.16 & 1893 \\ 
100\%  & 66.67 & 0 & 47.94 & 0 \\ 
\bottomrule
\end{tabularx}}
\label{tab:tlinks}
\end{table}

{\color{black}We examine the ability of our methodology to adapt to networks with time varying links (i.e., changing adjacency matrices) in Table~\ref{tab:tlinks}. 
This experiment randomly deactivates links within the adjacency matrix $\textbf{A}$ such that at any given iteration $t$, each $i,j$-th link will be deactivated, i.e., $\textbf{A}_{i,j}(t) = 0$, with probability $0\%, 25\%, 50\%, 75\%,$ or $100\%$. 
From Table~\ref{tab:tlinks}, higher link failure rate does reduce the classification performance for the proposed methodology, which makes sense as there are fewer links available and thus fewer D2D data offloading opportunities. 
Nonetheless, the resulting accuracies in the scenarios involving time-varying links are quite close to the case with fixed links. 
For both MNIST and FMNIST, the gap between $0\%$ link failure rate and $75\%$ link failure rate is within $7\%$, approximately. 
That being said, the case with total link failures (i.e., $100\%$ link failure rate) has over $6\%$ reduced final classification accuracy from the case with $75\%$ link failure rate. 
Thus, we can see that D2D data offloading has substantial impact on the resulting system performance.}

% {\color{black}We examine the ability of our methodology to adapt to networks with time varying links (i.e., changing adjacency matrices) in Fig.~\ref{fig:tlinks_mnist_ovr} and Fig.~\ref{fig:tlinks_fmnist_ovr}. As can be seen,  time-varying links do reduce the classification performance and/or convergence rate of the proposed sampling with D2D data offloading methodology. This is because time-varying links may change the amount of available links with statistically important D2D data offloading opportunities, thereby reducing the quality of local ML model training and the overall global ML model performance. 
% Nonetheless, the resulting classification accuracies in the scenarios with time-varying links remains quite close that of fixed links. 
% Specifically, the proposed methodology under time-varying links falls within $8\%$ and $4\%$, for MNIST and FMNIST respectively, of the classification performance for our methodology with fixed links. 
% This is because large-scale networks have many possible links among devices. As a result, our optimization formulation can find alternative, but possibly less effective, D2D data offloading opportunities from unsampled to sampled devices.}

\section{{\color{black}Examining Integration of FedDrop~\cite{wen2022federated}}}
\label{app:feddrop}

\begin{figure}[t]
\centering
\includegraphics[width=.86\linewidth]{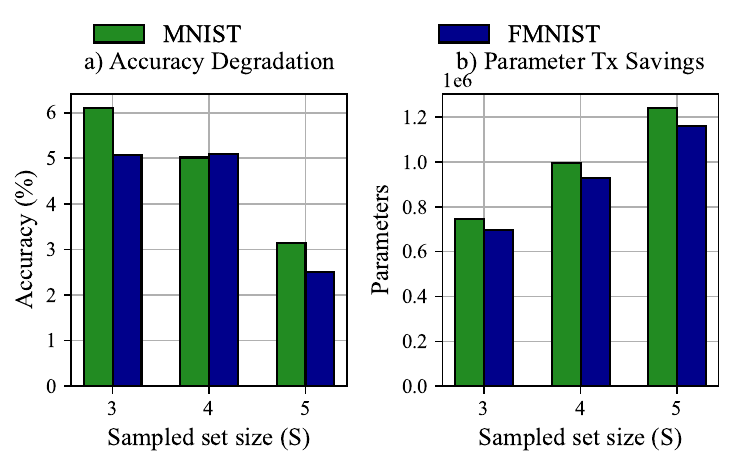}
\caption{{\color{black}Examining the effectiveness of FedDrop~\cite{wen2022federated} integrated with our proposed device sampling methodology on both MNIST and FMNIST datasets for a network with $N=100$. 
FedDrop has yields a small decrease in the final classification accuracies as well as a significant decrease in the number of parameters that are transmitted from the sampled devices to the server.}} 
\label{fig:feddrop_ovr}
\end{figure}

{\color{black}In the following, we examine the impact of federated dropout (FedDrop from~\cite{wen2022federated}), a methodology designed to dropout model parameters from the fully connected layers of ML models in federated settings, on our device sampling methodology. 
We use a dropout rate of $75\%$ for MNIST and $70\%$ for FMNIST on the fully connected layers of our CNNs. 
As our experiments rely on CNNs with $21840$ total parameters of which $16560$ parameters are for the fully connected layers, these dropout settings yield device-to-server ML model transmission savings of roughly $43\%$ and $46\%$ for MNIST and FMNIST, respectively. 
The classification performance results in Fig.~\ref{fig:feddrop_ovr} show that FedDrop integrated with our methodology yields a substantial decrease in total ML model parameters transmitted across all global aggregations. 
Moreover, these substantial resource savings only induce a minor decrease, under roughly $6\%$, in the final classification accuracies. 
Thus, in federated scenarios with device-to-server communication limitations, our methodology combined with FedDrop offers significant communication overhead resource savings.}

% \newpage
% \pagebreak
%test references length
%\newpage

% \balance
\bibliographystyle{IEEEtran}
\bibliography{References}

\vspace{-10mm}
\begin{IEEEbiography}[{\includegraphics[width=1in,height=1.25in,clip,keepaspectratio]{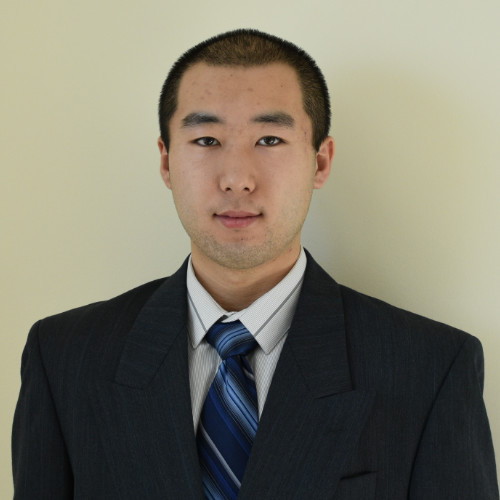}}]
{(Henry) Su Wang} received the B.S. (with distinction) and Ph.D. degrees in Electrical and Computer Engineering from Purdue University, West Lafayette in 2018 and 2023, respectively. 
He is currently a Postdoctoral Research Associate in the Department of Electrical and Computer Engineering at Princeton University. 
His research interests lie at the intersection of machine intelligence and networking. 
\end{IEEEbiography}

\vspace{-10mm}
\begin{IEEEbiography}[{\includegraphics[width=1in,height=1.25in,clip,keepaspectratio]{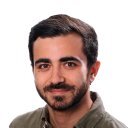}}]
{Roberto Morabito} (Member, IEEE) is an Assistant Professor in the Communication Systems Department at EURECOM, France.  His work intersects IoT, Edge Computing, and Distributed AI, focusing on trade-offs in AI service provisioning and orchestration under computing and networking resource constraints. Dr. Morabito earned his PhD in Networking Technology from Aalto University, Finland, in May 2019. From June 2019 to March 2021, he served as a Postdoctoral Researcher at the EDGE LAB, School of Electrical and Computer Engineering, Princeton University, USA.
\end{IEEEbiography}

\vspace{-10mm}
\begin{IEEEbiography}[{\includegraphics[width=1in,height=1.25in,clip,keepaspectratio]{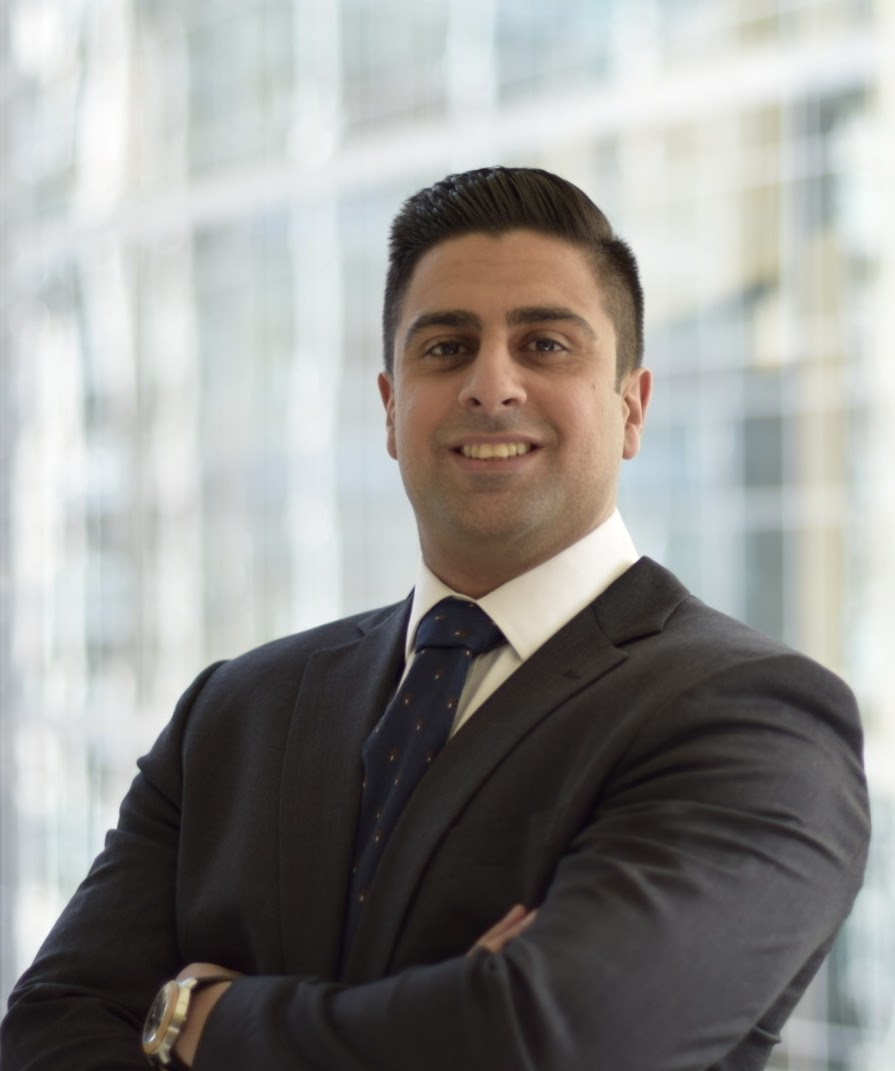}}]
{Seyyedali Hosseinalipour} (Member, IEEE) received
the B.S. degree in electrical engineering from Amirkabir University of Technology, Tehran, Iran, in 2015. He then received the M.S. and Ph.D. degrees in electrical engineering from North Carolina State University, NC, USA, in 2017 and 2020, respectively. He is currently an assistant professor at the Department of Electrical Engineering at the University at Buffalo (SUNY). His research interests include the analysis of modern wireless networks, synergies between machine learning methods and fog computing systems, distributed machine learning, and network optimization.
\end{IEEEbiography}

\vspace{-10mm}
\begin{IEEEbiography}[{\includegraphics[width=1in,height=1.25in,clip,keepaspectratio]{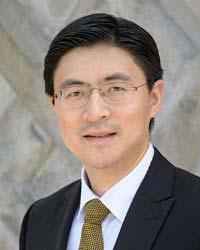}}]
{Mung Chiang} (F'12) is the President and Roscoe H. George Distinguished Professor of Electrical and Computer Engineering at Purdue University, and previously Arthur LeGrand Doty Professor at Princeton University and founded the Princeton Edge Lab. He received the 2013 Alan T. Waterman Award, and is a member of the American Academy of Arts and Sciences, National Academy of Inventors, and Royal Swedish Academy of Engineering Sciences.
\end{IEEEbiography}

\vspace{-10mm}
\begin{IEEEbiography}[{\includegraphics[width=1in,height=1.25in,clip,keepaspectratio]{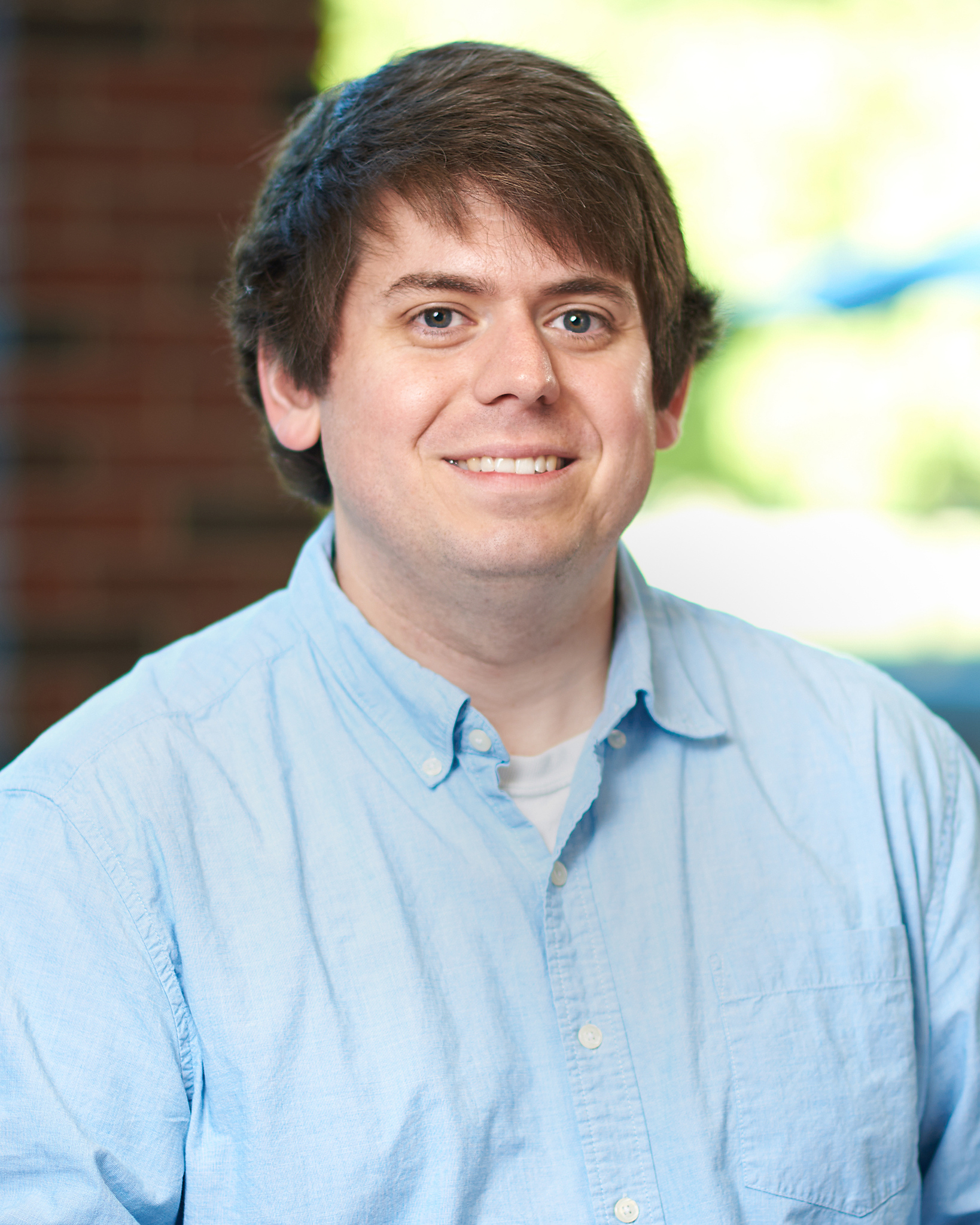}}]
{Christopher G. Brinton} (S’08, M’16, SM’20) is the Elmore Rising Star Assistant Professor of Electrical and Computer Engineering (ECE) at Purdue University. His research interest is at the intersection of networking, communications, and machine learning, specifically in fog/edge network intelligence, distributed machine learning, and AI/ML-inspired wireless network optimization. Dr. Brinton received the PhD and MS Degrees from Princeton in 2016 and 2013, respectively, both in Electrical Engineering.
\end{IEEEbiography}

\end{document}